\documentclass[11pt,a4paper]{article}
\usepackage[utf8]{inputenc}
\usepackage[T1]{fontenc}
\usepackage{amsmath,amssymb,amsthm,mathrsfs,bm}
\usepackage{geometry}
\usepackage{graphicx}
\usepackage{booktabs}
\usepackage{pgfplots}\pgfplotsset{compat=1.16}
\usepackage{hyperref}\hypersetup{colorlinks=true,linkcolor=blue,citecolor=blue,urlcolor=blue}
\usepackage{xfrac}

\newtheorem{proposition}{Proposition}

\newtheorem{remark}{Remark}

\newcommand{\R}{\mathcal{R}}
\newcommand{\Rc}{\mathcal{R}_{c}}
\newcommand{\N}{\mathcal{N}}
\newcommand{\Hs}{\mathcal{H}}
\newcommand{\su}{\mathfrak{su}(1,1)}
\newcommand{\sutwo}{\mathfrak{su}(2)}
\newcommand{\utwo}{\mathfrak{u}(2)}
\newcommand{\osp}{\mathfrak{osp}(1|2)}
\newcommand{\so}{\mathfrak{so}(2,1)}
\newcommand{\halg}{\mathfrak{h}_1}
\newcommand{\ralg}{\mathfrak{r}_1}
\newcommand{\jalg}{\mathfrak{j}_1}

\newcommand{\bialg}{\mathfrak{bi}_1}
\newcommand{\ket}[1]{|#1\rangle}
\newcommand{\bra}[1]{\langle#1|}
\newcommand{\ii}{\mathrm{i}}

\newcommand{\half}{\tfrac12}

\newcommand{\noleft}{\mathopen{}\mathclose\bgroup\left}
\newcommand{\noright}{\aftergroup\egroup\right}

\DeclareFontFamily{OMX}{yhex}{}
\DeclareFontShape{OMX}{yhex}{m}{n}{<->yhcmex10}{}
\DeclareSymbolFont{yhlargesymbols}{OMX}{yhex}{m}{n}
\DeclareMathAccent{\wideparen}{\mathord}{yhlargesymbols}{"F3}

\makeatletter
\let\wideparen@orig\wideparen
\renewcommand{\wideparen}[1]{%
  \mathchoice
    {\vphantom{\rule{0pt}{2.4ex}}\smash{\wideparen@orig{#1}}}% \displaystyle
    {\vphantom{\rule{0pt}{2.4ex}}\smash{\wideparen@orig{#1}}}% \textstyle
    {\wideparen@orig{#1}}% \scriptstyle
    {\wideparen@orig{#1}}% \scriptscriptstyle
}
\makeatother

\definecolor{MJRVpurple}{HTML}{7B00B5}

\definecolor{royalblue}{RGB}{65,105,225}
\definecolor{coralred}{RGB}{255,64,64}

\title{\bf The para-Krawtchouk oscillator and\\ its continuum limit}
\author{
St\'ephane Z. Beaulac\textsuperscript{1}\thanks{E-mail: stephane.jr.beaulac@umontreal.ca},\quad Nicolas Cramp\'e\textsuperscript{1,2}\thanks{E-mail: nicolas.crampe@cnrs.fr},\quad Quentin Labriet\textsuperscript{1}\thanks{E-mail: quentin.labriet@umontreal.ca},\\[3pt]
Lucia Morey\textsuperscript{1}\thanks{E-mail: lucia.morey@umontreal.ca},\quad Marlon Josue Rivera Valladares\textsuperscript{1}\thanks{E-mail: marlon.josue.rivera.valladares@umontreal.ca},\quad Luc Vinet\textsuperscript{1}\thanks{E-mail: luc.vinet@umontreal.ca}\\[6pt]
\small\textsuperscript{1}Centre de recherches math\'ematiques, Universit\'e de Montr\'eal, P.O.\ Box 6128,\\
\small Centre-ville Station, Montr\'eal (Qu\'ebec), H3C 3J7, Canada\\[2pt]
\small\textsuperscript{2}Laboratoire d'Annecy de Physique Th\'eorique, 9 Chemin de Bellevue, BP 110,\\
\small Annecy-le-Vieux, F-74941 Annecy Cedex, France
}
\date{\today}

\begin{document}
\maketitle

\begin{abstract}
The para-Krawtchouk polynomials define, through their recurrence relation, a finite
oscillator model that exhibits fractional revival. We determine the continuum limit of this
Hamiltonian, both from its recurrence relation and directly on the polynomials. The limit is
the singular (isotonic) oscillator $-\partial_\eta^2+\eta^2+\tfrac{\gamma^2-1}{4\eta^2}$. The
deformation parameter $\gamma$ survives as a central impurity of the couplings, which decays
as the inverse square of the distance. This is the marginal rate, which the continuum scaling
turns into a finite $\sfrac{1}{\eta^2}$ term. Using the decomposition of the para-Krawtchouk
polynomials into two Hahn families, we prove the confluence, for every degree, of the
eigenvectors on the two sub-lattices to generalized Hermite functions. The bispectral pair of
the discrete model, which generates the Hahn algebra, becomes in the limit a pair of
generators of $\su$, the dynamical algebra of the singular oscillator.
\end{abstract}

\section{Introduction}
\label{sec:intro}

This paper is the first of three companion papers that determine the continuum limits of the
discrete Hamiltonians attached to three families of polynomials orthogonal on
bi-lattices: the para-Krawtchouk polynomials \cite{VZ}, treated here, and the para-Racah \cite{paraRacah}
and para-Bannai--Ito \cite{PBI} polynomials, treated in \cite{PRcont,PBIcont}, whose properties are reviewed in \cite{BCLMNV}. Each of these Hamiltonians is the Jacobi matrix $J$ of the family, that is, the
symmetric tridiagonal matrix that represents the three-term recurrence relation of the
orthonormal polynomials $p_n$, $x\,p_n=\sqrt{U_{n+1}}\,p_{n+1}+B_n\,p_n+\sqrt{U_n}\,p_{n-1}$. 

Such a matrix can be viewed as the one-excitation restriction of an $XX$ spin chain whose
nearest-neighbour couplings and local magnetic fields are engineered to be $\sqrt{U_n}$ and
$B_n$ respectively. The choice of the polynomials then determines the transport of an excitation along the
chain. Perfect state transfer (PST), in which the excitation reappears intact at the far end,
occurs for the Krawtchouk Hamiltonian, whose continuum limit is the harmonic oscillator, the
Krawtchouk polynomials tending to the Hermite polynomials. Fractional revival (FR) is the
more general phenomenon in which the excitation reappears as a coherent superposition
localized at the two ends, an entangled state \cite{GVZ}. PST is the special case of FR in
which the whole weight lands at the far end. More generally, the one-parameter evolution
$e^{-\ii\theta J}$ is a harmonic (periodic or quasi-periodic) motion, of which PST and FR are
particular instants. The para-Krawtchouk Hamiltonian \cite{VZ} is an analytically solvable
model realizing this whole family, tuned by a single parameter $\gamma\in(0,2)$.%
\footnote{The prefix records that the bi-lattice $\{2s\}\cup\{2s+\gamma\}$ of
\S\ref{sec:chain} is the spectrum of $a^\dagger a$ for the parabose oscillator of parameter
$\tfrac{\gamma-1}{2}$ \cite{Rosenblum}. It is unrelated to the para-orthogonal polynomials
on the unit circle of \cite{JNT}, the combinations $\Phi_n(z)+\tau\Phi_n^{*}(z)$,
$|\tau|=1$, used in Szeg\H{o} quadrature.}

This paper answers one question: what is the continuum limit of the single Hamiltonian $J$
that defines this oscillator? The answer is the singular (isotonic, or Calogero) oscillator.
Three features make the result more than a routine contraction, and each requires a separate
analysis.

First, $\gamma$ cannot disappear in the limit: the discrete model depends on it (its whole spectrum
does), and one must recover the harmonic oscillator precisely at $\gamma=1$. The mechanism is
a central impurity: the para-Krawtchouk Hamiltonian is the Krawtchouk one perturbed by a
defect concentrated at the middle and decaying as the inverse square of the distance. This is a marginal perturbation, which neither localizes nor washes out, and which
becomes, in the continuum coordinate $\eta$ of the oscillator, the $\sfrac{1}{\eta^2}$ centrifugal
term whose coefficient carries $\gamma$.

Second, the limiting object is not the parabosonic (Wigner--Dunkl) oscillator
\cite{Wigner} that the family resemblance suggests. The Wigner--Dunkl oscillator has a reflection-graded
potential $\mu(\mu-\R)/\eta^2$, $\R$ being the reflection $\eta\mapsto-\eta$, whose two parity
sectors carry centrifugal exponents differing by one unit. The para-Krawtchouk limit has instead a scalar $\sfrac{1}{\eta^2}$ coefficient,
common to both sectors, and the reflection enters only through the choice of boundary
behaviour at the origin: the even and odd sub-lattices realize the two independent solutions
of one and the same singular oscillator, with exponents $(1\mp\gamma)/2$ separated not by $1$
but by $\gamma$. This separation is exactly the fractional--revival phase. The two
descriptions coincide only at $\gamma=1$.

Third, the singularity at the origin forces a choice of self-adjoint extension. The same
delicate question arises for the radial Schr\"odinger equation of a particle in the field of
a magnetic monopole or in a $1/r^2$ potential \cite{DHV}. The impurity of the discrete model
selects that extension, and it is what distinguishes the two channels.

The computation is carried out twice, in parallel with the classical Krawtchouk case.
Section~\ref{sec:chain} defines the polynomials, the Hamiltonian, its spectrum and its
eigenvectors. In \S\ref{sec:eqlimit} we contract the recurrence relation, read as a discrete
Schr\"odinger equation, and identify the central impurity of the couplings that survives as
the $\sfrac{1}{\eta^2}$ term. In \S\ref{sec:polylimit} we take the limit directly on the
polynomials, using their explicit hypergeometric form and the reduction of the
para-Krawtchouk polynomials to two Hahn families \cite{BCLMNV}, and prove the confluence,
for every degree, to the generalized Hermite functions (Proposition~\ref{prop:conv}), the
envelope being read off the normalization of the polynomials. Section~\ref{sec:reflection}
treats the reflection and the two channels, \S\ref{sec:algebra} what the bispectral pair of
the discrete model, which generates the Hahn algebra, becomes in the limit, and
\S\ref{sec:even} the case of even $N$, in which the impurity is shared between the couplings
and the fields and which has the same limit. Sections \ref{sec:extension}--\ref{sec:doubling}
treat the self--adjoint extension at the singular point, the comparison with the Krawtchouk
and Hahn oscillators and with the $\su$ Meixner oscillator on the semi--infinite lattice,
which reaches the same singular oscillator with no contraction at all, and the contrast with
the doubling construction of Oste and Van der Jeugt. Section~\ref{sec:FR} turns to fractional revival and
the harmonic evolution, at finite $N$ and in the limit. Appendix~\ref{app:contraction}
establishes the sign structure of the eigenvectors on which the contraction of
\S\ref{sec:eqlimit} rests, from the oscillation theorem of Gantmacher and Krein.

\section{The para-Krawtchouk polynomials}
\label{sec:chain}
Let $N=2j+p$, where $p=0,1$ according as $N$ is even or odd, respectively, and $j$ is a
positive integer. The monic para-Krawtchouk polynomials $P_n(x)$ are defined by the three-term recurrence relation
\begin{equation}
xP_n(x)=P_{n+1}(x)+B_nP_n(x)+U_nP_{n-1}(x),
\qquad n=0,1,\ldots,N,
\label{monicPKrecurrence}
\end{equation}
with $ P_{-1}(x)=0,$ $ P_0(x)=1$ and the coefficients given by
\begin{align}
B_n
&=-\frac{(2j+p-n)(2j+2p-2-2n+\gamma)}{2(2n-2j-2p+1)}-\frac{n(2j+2-2n-\gamma)}{2(2n-2j-1)}\,,\label{eq:recurrenceB}\\
U_n
&=\frac{n\left(2 j +p -n +1\right) \left(2 j +2 p -2 n +\gamma \right)  \left(2 j +2-2 n -\gamma \right)}{4 \left(2 n -1-2 j -2 p \right) \left(2 n -2 j -1\right)}\,.\label{eq:recurrenceU}
\end{align}
The orthonormal polynomials are
$p_n=P_n/\sqrt{h_n}$ with $h_n=U_1U_2\cdots U_n$ ($h_0=1$), and \eqref{monicPKrecurrence} becomes
\begin{equation}
\label{eq:recnorm}
xp_n(x)=\sqrt{U_{n+1}}\,p_{n+1}(x)+B_n\,p_n(x)+\sqrt{U_n}\,p_{n-1}(x).
\end{equation}
Let $\ket{n}$, $n=0,\dots,N$, be the canonical basis of $\mathbb C^{N+1}$ and $J$ the Jacobi matrix of the family, i.e. the real symmetric tridiagonal
matrix with
\begin{align}
    &\langle n\,|J|\,n\rangle=B_n\,,\label{eq:defn-J-B}\\
    &\langle {n}\,|J|\,{n+1}\rangle=\langle {n+1}\,|J|\,{n}\rangle=\sqrt{U_{n+1}}\label{eq:defn-J-U}\,.
\end{align}
Its
eigenvalues are the zeros $x_0<x_1<\dots<x_N$ of $P_{N+1}$ (defined by \eqref{monicPKrecurrence} with $n=N$). They were determined in \cite{VZ} and form the bi-lattice
\begin{equation}
    \label{eq:bilattice}
\begin{aligned}
&x_{2s}=2s,\qquad
s=0,\ldots,j,\\
&x_{2s+1}=2s+\gamma,
\qquad
s=0,\ldots,j-1+p,
\end{aligned}
\end{equation}
that is $\{0,\gamma,2,2+\gamma,4,\dots\}$: two interleaved uniform lattices offset by
$\gamma$.
For the two sub-lattices to interlace without colliding, $\gamma$ must lie in the range $\gamma \in (0,2). $
At $\gamma=1$
the bi-lattice is uniform, $x_s=s$, and \eqref{monicPKrecurrence} reduces to the symmetric Krawtchouk
recurrence \cite{KLS}. By \eqref{eq:recnorm}, the normalized eigenvectors $\ket{x_s}$ of $J$, with
   $J\ket{x_s}=x_s\ket{x_s}$, are related to the canonical basis by
\begin{align}
&\ket{x_s}=\sum_{n=0}^N\sqrt{w_s}\,p_n(x_s)\,\ket{n}\,,\\
&\ket{n}=\sum_{s=0}^N\sqrt{w_s}\,p_n(x_s)\,\ket{x_s}\,,\label{eq:eigvec}
\end{align}
where the $w_s>0$ are the weights of the discrete orthogonality relation
\begin{equation}
    \sum_{s=0}^Nw_s\,p_n(x_s)p_{n'}(x_s)=\delta_{nn'}\,,
\end{equation}
normalized by $\sum_sw_s=1$; and for $N=2j+1$ They read \cite{VZ,BCLMNV}:
\begin{equation}
w_{2s}
=
\frac{(1-\gamma/2)_j}{2^N(1/2)_j}
\,
\frac{(-j)_s(-\gamma/2-j+1-p)_s}
{s!(1-\gamma/2)_s},
\qquad
s=0,1,\ldots,j,
\label{weight_even}
\end{equation}
\begin{equation}
w_{2s+1}
=(-1)^{p+1}
\frac{(1+\gamma/2)_{j-1+p}}{2^N(1/2)_j}
\,
\frac{(-j)_{s+1-p}(\gamma/2-j)_s}
{s!(1+\gamma/2)_s},
\qquad
s=0,1,\ldots,j-1+p.
\label{weight_odd}
\end{equation}
The weights for the even case $N=2j$ can be derived through the Christoffel transform \cite{VZ} from those of the odd case. 
Finally, the wavefunctions are defined by
\begin{equation}
\label{eq:def-un}
    u_s(n):=\langle n|x_s\rangle=\sqrt{w_s}\,p_n(x_s)\,.
\end{equation}

\begin{remark}
The eigenvalues of $J$ are the first points of the two infinite lattices
$\{2s\}_{s\ge 0}$ and $\{2s+\gamma\}_{s\ge 0}$, namely $\{2s\}_{s=0}^{j}$ and
$\{2s+\gamma\}_{s=0}^{j-1+p}$. Since $N$ does not enter these formulas,
increasing $N$ only adds eigenvalues at the top of the spectrum and leaves the rest fixed.
The limit $N\to\infty$ can therefore be taken eigenvalue by eigenvalue. 
\end{remark}

\paragraph{The site reversal.}
Let $\Rc$ be the site reversal, the unitary involution
$\Rc\ket{n}=\ket{{N-n}}$. The
matrix $J$ is persymmetric, $\R_c J\R_c=J$, so $\R_c$ commutes with $J$. Since the spectrum of $J$ is
simple, every eigenvector $\ket{x_s}$ is therefore also an eigenvector of $\R_c$,
\begin{equation}
\R_c\ket{x_s}=\epsilon_s\ket{x_s},\qquad\epsilon_s=\pm1.
\end{equation}
Reading $\R_c\ket{x_s}=\epsilon_s\ket{x_s}$ at the site $\ket{0}$ gives $p_N(x_s)=\epsilon_s$, and
$p_N$ takes alternating signs at the $N+1$ points $x_s$, which interlace its $N$ zeros, with
$p_N(x_N)>0$ because all its zeros lie below $x_N$ and its leading coefficient is positive.
Hence
\begin{equation}
\label{eq:eps}
\epsilon_s=(-1)^{N+s}.
\end{equation}
 On the bi-lattice \eqref{eq:bilattice}, $\epsilon_s$ is constant on each sub-lattice. Indeed,  writing
$N=2j+p$, \eqref{eq:eps} gives
\begin{equation}
\label{eq:eps-sub-lattice}
\epsilon_{2s}=(-1)^{p},
\qquad
\epsilon_{2s+1}=(-1)^{p+1}.
\end{equation}

\paragraph{The explicit form of the polynomials.}
For $n\le j$ the para-Krawtchouk polynomials are
\begin{equation}
P_n(x)=\kappa_n\,
{}_3F_2\!\left(
\begin{matrix}
-n,\; n-2j-p,\; -x/2 \\[0.2cm]
-j,\; -j+1-p-\gamma/2
\end{matrix}
\,;1
\right),\qquad \kappa_n=\dfrac{2^n(-j)_n\left(-j+1-p-\gamma/2\right)_n}
{(n-2j-p)_n}.
\label{eq:P3F2}
\end{equation}
This expression is obtained in \cite{sklyanin} through a non-standard truncation of the continuous Hahn polynomials, and is
equivalent to the expression of \cite[\S4]{VZ} in terms of the complementary Bannai--Ito
polynomials.  The sites $n>j$ need not be treated separately since the site reversal relates them to
$n\le j$. Reading $\R_c\ket{x_s}=\epsilon_s\ket{x_s}$ site by site gives
\begin{equation}
\label{eq:mirror}
p_{N-n}(x_s)=\epsilon_s\,p_n(x_s),
\end{equation}
so that the polynomials at the sites $n>j$ are those at $N-n\le j$, up to the sign
$\epsilon_s=(-1)^{N+s}$.

\section{The recurrence as a discrete Schr\"odinger equation}
\label{sec:eqlimit}

We will first derive the continuum limit of the recurrence relation for $N$ odd ($p=1$). The case of $N$ even is treated in \S\ref{sec:even}. The starting
point is that the recurrence relation \eqref{eq:recnorm}, evaluated at a zero
$x_s$ of $P_{N+1}$, is a discrete, time-independent Schr\"odinger
equation for $\bm{p}\left(x\right)=\left[p_0\left(x\right),...,p_N\left(x\right)\right]^{\bm{\top}}$,
\begin{equation}
\label{eq:discreteSchr}
J\,\bm{p}\left(x_s\right)=x_s\,\bm{p}\left(x_s\right),
\end{equation}
with the Jacobi matrix $J$ acting as the Hamiltonian, the degree $n$ as the discrete position, and $x_s$ as the energy. It is convenient to replace
$N$ by $j$, and $n$ by its rescaled distance from the centre
\begin{equation}
    \label{eq:etanj}
    \eta:=\frac{n-j-1}{\sqrt{j+1}}\,,
\end{equation}
so that $n=j+1+\sqrt{j+1}\,\eta\,.$ 
The recurrence coefficients \eqref{eq:recurrenceB} and \eqref{eq:recurrenceU} become:
\begin{equation}
\label{eq:coefficientsMeta}
B_{n}=j+\frac{\gamma}{2},\quad
\sqrt{U_{n}}=\frac{(j+1)}{2}\,{\frac{\left(1-\frac{\eta^2}{(j+1)}\right)^{\frac{1}{2}}\left(1-\frac{\gamma^2}{4(j+1)\eta^2}\right)^{\frac{1}{2}}}{\left(1-\frac{1}{4(j+1)\eta^2}\right)^{\frac{1}{2}}}}.
\end{equation}
These expressions are exact, but in the limit of large $j$ and fixed $\eta\neq0$ (i.e., $|\eta|\gg\frac{1}{\sqrt{j+1}}$), we can use the binomial expansion to get the leading terms in $\sqrt{U_n}$:
\begin{equation}
\label{eq:UnMetaExpanded}
\sqrt{U_{n}}=\frac{(j+1)}{2}+\frac{1-\gamma^2}{16\eta^2}-\frac{\eta^2}{4}+O\noleft(\frac{1}{j}\noright)
\end{equation}
We can also note that in this limit, the mirror map $\eta\mapsto-\eta-\frac{1}{\sqrt{j+1}}$ (equivalent to $n\mapsto N-n$) becomes $\eta\mapsto-\eta$.

\paragraph{The staggered ansatz.}
Taking a fixed $s$ associated with the eigenvalue $x_s$, then in the sequence of $\bm{p}\noleft(x_s\noright)$ (i.e., $p_0(x_s),\dots,p_N(x_s)$), there will be $N-s$ sign changes (shown in Appendix~\ref{app:contraction} via \cite[Ch. II.1]{GantmacherKreinBook}). Conversely, the sequence of envelopes $\varphi\noleft(\eta\noright)$, which we define as
\begin{equation}
\label{eq:envelopePhi}
\varphi\noleft(\eta\noright):=\noleft(-1\noright)^{n}\,p_{n}\noleft(x_s\noright),
\end{equation}
is specifically constructed to change signs $s$ times, and make these changes more apparent. This can then be inverted to get a general form for the $p_n\noleft(x_s\noright)=(-1)^{n}\,\varphi\noleft(\eta\noright)$. In terms of $\varphi$, the recurrence \eqref{eq:recnorm} reads
\begin{equation}
\label{eq:recPhi}
-\sqrt{U_{n+1}}\,\varphi\noleft(\eta+(j+1)^{-1/2}\noright)+B_n\,\varphi\noleft(\eta\noright)-\sqrt{U_n}\,\varphi\noleft(\eta-(j+1)^{-1/2}\noright)=x_s\varphi\noleft(\eta\noright),
\end{equation}
where we have used the mapping $n+m\mapsto\eta+m\,(j+1)^{-1/2}$.
Now, again in the limit that $\noleft\lvert\eta\noright\rvert\gg (j+1)^{-1/2}$, we can Taylor expand $\varphi\noleft(\eta\pm (j+1)^{-1/2}\noright)$ around $\eta$ to get:
\begin{equation}
\label{eq:expansionPhi}
\varphi\noleft(\eta\pm (j+1)^{-1/2}\noright)=\varphi\noleft(\eta\noright)\pm\frac{1}{\sqrt{j+1}}\,\partial_{\eta}\varphi\noleft(\eta\noright)+\frac{1}{2(j+1)}\,\partial^2_{\eta}\varphi\noleft(\eta\noright)+O\noleft(\frac{1}{(j+1)^{3/2}}\noright).
\end{equation}
Combining this with \eqref{eq:coefficientsMeta}, \eqref{eq:UnMetaExpanded} and \eqref{eq:recPhi}, while assuming a low-lying eigenvalue ($s\ll j$) to match the orders on both sides of the recurrence relation, finally gives to leading order:
\begin{equation}
\label{eq:singosc}
\Big[-\partial_\eta^2+\eta^2+\frac{\gamma^2-1}{4\eta^2}\Big]\varphi
=\big(2x_s+2-\gamma\big)\varphi ,
\end{equation}
where the terms proportional to $j$ in $B_n$, $\sqrt{U_n}$ and $\sqrt{U_{n+1}}$, and the first derivative terms cancelled out exactly. In this form, we can directly see that taking $N\to\infty$ leads to a standard harmonic confinement, $-\partial_\eta^2+\eta^2$, as recovered in the limit continuum of the Krawtchouk Hamiltonian, perturbed by a central singular impurity, $\frac{\gamma^2-1}{4\,\eta^2}$, which vanishes when $\gamma=1$. Hence, we fully recover the Krawtchouk limit for the appropriate choice of $\gamma$. 

\begin{remark}
Equation \eqref{eq:singosc} is the eigenvalue equation of the singular, or isotonic, oscillator
\begin{equation}
\label{eq:Hdef}
\mathcal{H}=-\partial_\eta^2+\eta^2+\frac{\gamma^2-1}{4\eta^2}.
\end{equation}
The point $\eta=0$ is a regular singular point, with
indicial equation
\begin{equation}
\label{eq:indicial}
\sigma(\sigma-1)=\frac{\gamma^2-1}{4},
\end{equation}
whose roots are $\sigma_\mp=\frac{1\mp\gamma}{2}$. Near the origin, the
solutions are therefore combinations of
$|\eta|^{\sigma_\mp}\,q_\mp(\eta^2),$
where $q_\mp$ is a power series with $q_\mp(0)\neq0$. Since the
exponents differ by $\gamma\notin2\mathbb{Z}$, no logarithmic terms arise.
The square-integrable eigenfunctions obtained in the limit are given
explicitly in Proposition~\ref{prop:conv}.  
\end{remark}

\paragraph{The central impurity.}
From \eqref{eq:UnMetaExpanded}, it is clear that the singular term $\frac{\gamma^2-1}{4\eta^2}$ arises from the large $N$ limit of the $U_n$. However, the emergence of this extra term is not merely a result of the continuum limit; the difference with the standard Krawtchouk chain is already visible at the discrete level. Indeed, if we take the ratio of the para-Krawtchouk $U_n$ with the Krawtchouk couplings $U_n^{K}=\tfrac{n}{4}(2j+2-n)$, we find:
\begin{equation}
\label{eq:impurity}
\frac{U_n}{U_n^{K}}=\frac{(2j+2-2n)^2-\gamma^2}{(2j+1-2n)(2j+3-2n)}
=\frac{4\left(j+1\right)\eta^2-\gamma^2}{4\left(j+1\right)\eta^2-1}.
\end{equation}
Thus, 
\begin{equation}
    U_n=U_n^{K}\left(1-\frac{\gamma^2-1}{4(j+1)\eta^2}+O\noleft(\frac{1}{\noleft(j+1\noright)^2\eta^4}\noright)\right)\,.
\end{equation} 
Therefore, we note again that the para-Krawtchouk chain is the standard chain with an added defect at its centre.

As was previously mentioned, we know that, for the Krawtchouk case (i.e., $\gamma=1$), we have a perfect state transfer.
For $\gamma\neq1$ however, the impurity converts perfect transfer
into fractional revival (\S\ref{sec:FR}), and, being marginal, it is exactly this impurity that survives as the
$\sfrac{1}{\eta^2}$ term in the continuum potential. 
\begin{remark}
$\gamma$ survives precisely because the impurity is marginal. A defect decaying faster
than $\sfrac{1}{\eta^2}$ would be irrelevant (giving the ordinary oscillator for all $\gamma$); one
decaying slower would be relevant (destroying the oscillator). The inverse--square law is the
borderline that produces a finite $\sfrac{1}{\eta^2}$ potential.
\end{remark}
Figure~\ref{fig:profiles} compares the ratio \eqref{eq:impurity} with the continuum potential
$V(\eta)=\eta^2+\tfrac{\gamma^2-1}{4\eta^2}$. We note that for $\gamma<1$ the couplings exceed
the Krawtchouk ones on every point (except at $\eta=0$) by $\tfrac{1-\gamma^2}{4(j+1)\eta^2-1}$. This
excess, decaying as the inverse square of the distance, is the microscopic image of the continuum
attractive $\sfrac{1}{\eta^2}$ term. Conversely, for $\gamma>1$ the corresponding
deficit is the image of the repulsive term. In both the discrete and continuum regimes, these flatten to the clean case at $\gamma=1$.
For $\eta=0$, the discrete trend is reversed with $U_{j+1}/U_{j+1}^{K}=\gamma^2$: for $\gamma<1$, there is a deficit; and for $\gamma>1$, an excess. In the continuum limit, the bulk equation diverges at the centre. This is precisely where the discrete and continuum regimes decouple. As the potential itself is produced entirely by the tail of the impurity, it does not carry any information about the behaviour of the eigenfunctions at the centre. As such, the bulk equation, which is singular at that point, does not fix the behaviour of the eigenfunctions there. Because of this, we must supply the differential equation with a boundary condition at $\eta=0$.
\begin{figure}[ht!]
\centering
\begin{tikzpicture}
\begin{axis}[
  width=0.5\textwidth,height=6.2cm,
  xlabel={position $n$},ylabel={$U_n/U_n^{K}$},
  title={\small(a) discrete couplings},
  xmin=1,xmax=41,ymin=0,ymax=2.5,
  legend style={at={(0.97,0.97)},anchor=north east,draw=none,font=\footnotesize},
  tick label style={font=\footnotesize},label style={font=\small},
]
  \addplot[royalblue,thick,mark=*,mark size=0.9pt,samples=41,domain=1:41]
    {((42-2*x)^2-0.25)/((41-2*x)*(43-2*x))};
  \addlegendentry{$\gamma=0.5$}
  \addplot[black,dashed,thick,samples=41,domain=1:41]{1};
  \addlegendentry{$\gamma=1$}
  \addplot[coralred,thick,mark=triangle*,mark size=1.1pt,samples=41,domain=1:41]
    {((42-2*x)^2-2.25)/((41-2*x)*(43-2*x))};
  \addlegendentry{$\gamma=1.5$}
\end{axis}
\end{tikzpicture}\hfill
\begin{tikzpicture}
\begin{axis}[
  width=0.5\textwidth,height=6.2cm,
  xlabel={$\eta$},ylabel={$V(\eta)=\eta^2+\tfrac{\gamma^2-1}{4\eta^2}$},
  title={\small(b) continuum potential},
  xmin=-2.6,xmax=2.6,ymin=-3.5,ymax=7,
  legend style={at={(0.5,0.97)},anchor=north,draw=none,font=\footnotesize},
  tick label style={font=\footnotesize},label style={font=\small},
]
  \addplot[royalblue,thick,samples=140,domain=0.18:2.6,restrict y to domain=-4:8]
    {x^2+(0.25-1)/(4*x^2)};
  \addlegendentry{$\gamma=0.5$}
  \addplot[royalblue,thick,forget plot,samples=140,domain=-2.6:-0.18,restrict y to domain=-4:8]
    {x^2+(0.25-1)/(4*x^2)};
  \addplot[black,dashed,thick,samples=80,domain=-2.6:2.6]{x^2};
  \addlegendentry{$\gamma=1$}
  \addplot[coralred,thick,samples=160,domain=0.14:2.6,restrict y to domain=-4:8]
    {x^2+(2.25-1)/(4*x^2)};
  \addlegendentry{$\gamma=1.5$}
  \addplot[coralred,thick,forget plot,samples=160,domain=-2.6:-0.14,restrict y to domain=-4:8]
    {x^2+(2.25-1)/(4*x^2)};
\end{axis}
\end{tikzpicture}
\caption{(a) The discrete coupling ratio $U_n/U_n^K$ on the $N=41$ points of a lattice of
$N+1=42$ sites: the para-Krawtchouk couplings equal the Krawtchouk ones except for a
central impurity, decaying as $1-\tfrac{\gamma^2-1}{4(j+1)\eta^2}$ with the distance $|\eta|\sqrt{j+1}$ of the
position to the middle, absent at $\gamma=1$; the central point deviates in the opposite
direction ($U_{j+1}/U_{j+1}^K=\gamma^2$).
(b) The continuum potential $V(\eta)=\eta^2+\tfrac{\gamma^2-1}{4\eta^2}$ obtained in the
limit: attractive for $\gamma<1$, purely harmonic at $\gamma=1$, repulsive for $\gamma>1$.
The tail of the impurity of (a) becomes the $\sfrac{1}{\eta^2}$ term of (b); the boundary behaviour
at $\eta=0$ is decided in the central layer, where the tail expansion fails.}
\label{fig:profiles}
\end{figure}

To do this, we can reconsider the derivation of \eqref{eq:singosc}. Although this expression is valid for the bulk values of $n$ (and therefore of $\eta$), it is important to note that it breaks down at $n=j+1$ ($\eta=0)$. Indeed, the assumption that $\noleft\lvert\eta\noright\rvert\gg (j+1)^{-1/2}$ is no longer valid for this value, and hence the derivation around this centre is not consistent. However, we can directly compute $\sqrt{U_n}$ and $\sqrt{U_{n+1}}$ at this point. To leading order, these are:
\begin{equation}
    \label{eq:coefficUM}
    \sqrt{U_{j+1}}=\frac{(j+1)}{2}\,\gamma\qquad\sqrt{U_{j+2}}=\frac{(j+1)}{2}\,\sqrt{\frac{4-\gamma^2}{3}}+O\noleft(\frac{1}{j}\noright).
\end{equation}
If we then consider the values $\eta=\pm\frac{1}{2\sqrt{j+1}}$, which lie directly next to the centre, then to order $j$, \eqref{eq:recPhi} gives:
\begin{equation}
\label{eq:A:centre}
\varphi\big(\pm\tfrac1{2\sqrt {j+1}}\big)=\frac\gamma2\,\varphi\big(\mp\tfrac1{2\sqrt {j+1}}\big)
+\frac12\sqrt{\frac{4-\gamma^2}3}\;\varphi\big(\pm\tfrac3{2\sqrt {j+1}}\big).
\end{equation}
This is the matching condition that fixes the behaviour of $\varphi$ as it approaches $\eta=0$ (from both sides). 
% In other words, this picks the self--adjoint extension of \eqref{eq:singosc}, which the bulk equation leaves open. This caveat is specific to taking the limit of the discrete Schr\"odinger equation. This is not the case for taking the limit on the polynomials: that limit directly supplies the operator and its $\eta=0$ boundary behaviour (\S\ref{sec:extension}). We compute this limit in what follows.
For an envelope that is even (odd) under the mirror map,
\eqref{eq:A:centre} becomes
\begin{equation}
\label{eq:A:centre-parity}
\Big(1\mp\frac\gamma2\Big)\,\varphi\big(\tfrac1{2\sqrt{j+1}}\big)
=\frac12\sqrt{\frac{4-\gamma^2}3}\;\varphi\big(\tfrac3{2\sqrt{j+1}}\big),
\end{equation}
so that the centre of the chain acts differently on even and odd
envelopes. The resulting behaviour of $\varphi$ at $\eta=0$, i.e.\ the
self-adjoint extension of \eqref{eq:singosc}, is read off the polynomials
(\S\ref{sec:extension}), whose limit we consider in what follows.

\section{The continuum limit taken directly on the polynomials}
\label{sec:polylimit}

We now compute the limit of the eigenvectors directly from the polynomials, as in the classical limit from the Krawtchouk to the Hermite polynomials. We use the reduction of the para-Krawtchouk polynomials on each sub-lattice to a family of dual Hahn polynomials \cite{BCLMNV}. Throughout this section $N=2j+1$ and, as in \S\ref{sec:eqlimit}, $\eta=(n-j-1)/\sqrt{j+1}$.

\paragraph{The Krawtchouk case.}
For $\gamma=1$ the polynomials are the symmetric Krawtchouk polynomials
\begin{equation}
  K_n(x;\half,N)={}_2F_1\left(\begin{matrix}-n,-x\\-N\end{matrix};2\right)\,,
\end{equation}
orthogonal on $x=0,1,\dots,N$ with respect to the
binomial weight $w_x=2^{-N}\binom Nx$ and with norm $h_n=1/\binom Nn$ \cite[\S9.11]{KLS}.
Setting
$x\to \frac{N}{2}+\sqrt{\frac{N}{2}}\eta$, one has \cite[\S9.11]{KLS}
\begin{equation}
\lim_{N\to\infty} \sqrt{\binom{N}{n}}K_n\left(\frac{N}{2}+\eta\sqrt{\frac{N}{2}};\frac12,N\right)=\frac{(-1)^n}{\sqrt{2^nn!}}H_n(\eta)\,,
\end{equation}
and, by the De Moivre--Laplace theorem $w_x\sim\sqrt{2/(\pi N)}\,e^{-\eta^2}$. Together,
\begin{equation}
  \left(\frac{ N}{2}\right)^{1/4}\sqrt{\frac{w_x}{h_n}}\,K_n(x;\half,N)\longrightarrow \frac{(-1)^n}{\pi^{1/4}\sqrt{2^nn!}}\,e^{-\eta^2/2}H_n(\eta)\,,
\end{equation}
the normalized Hermite functions. Since $K_n(x;\half,N)=K_x(n;\half,N)$, the same limit holds
with the roles of $n$ and $x$ exchanged.

\paragraph{The two sub-lattices as dual Hahn polynomials.}
On the bi-lattice \eqref{eq:bilattice} the para-Krawtchouk polynomials can be written as \cite{BCLMNV}
\begin{equation}
\label{eq:Q}
P_n(x_{2t})=\kappa_n\,Q^e_n(t),\qquad
P_n(x_{2t+1})=\kappa_n\,\frac{(-j+\tfrac\gamma2)_n}{(-j-\tfrac\gamma2)_n}\,Q^o_n(t),\qquad  n\le j, 
\end{equation}
with $\kappa_n$ as in \eqref{eq:P3F2} and $Q_n^e$, $Q_n^o$ the terminating $_3F_2$ hypergeometric functions
\begin{equation}
\label{eq:Qdef}
Q^e_n(t)={}_3F_2\!\left(\begin{matrix}-t,\,-n,\,n-N\\
-j-\tfrac\gamma2,\,-j\end{matrix};1\right),\qquad
Q^o_n(t)={}_3F_2\!\left(\begin{matrix}-t,\,-n,\,n-N\\
-j+\tfrac\gamma2,\,-j\end{matrix};1\right)\,.
\end{equation}
These are dual Hahn polynomials \cite[\S9.6]{KLS}, defined by 
\begin{equation}
\label{eq:def-hahnpols}
  R_t(\lambda(n);\alpha,\beta,K)={}_3F_2\!\left(\begin{matrix} -t,-n,n+\alpha+\beta+1\\\alpha+1,-K\end{matrix};1\right) \,,\quad \lambda(n)=n(n+\alpha+\beta+1)\,,\quad K\in\mathbb{Z}_{>0}\,.
\end{equation}
With the parameters $(\alpha_e,\beta_e)=(-j-1-\tfrac\gamma2,\,-j-1+\tfrac\gamma2)$, for which $\alpha_e+\beta_e+1=-N$ we have
\begin{equation}
\label{eq:dualHahn}
Q^e_n(t)=R_t\big(\lambda(n);\,\alpha_e,\,\beta_e,\,j\big),\qquad
Q^o_n(t)=R_t\big(\lambda(n);\,\beta_e,\,\alpha_e,\,j\big)\,,\qquad \lambda(n)=n(n-N).
\end{equation}
For fixed $t$, $Q^e_n(t)$ and $Q^o_n(t)$ are thus polynomials of degree $t$ in
$\lambda(n)$, and in particular they are invariant under $n\mapsto N-n$. The two sub-lattices thus correspond to the dual Hahn family with $\alpha$ and $\beta$
exchanged, equivalently with $\gamma$ replaced by $-\gamma$. 

\paragraph{The confluence for general degree.}
Substituting $n=j+1+\sqrt{j+1}\,\eta$ and $ N-n=j-\sqrt{j+1}\,\eta$ into $\lambda(n)=n(n-N)$ gives
\begin{align}
\lambda(n)&=-(j+1+\sqrt{j+1}\,\eta)\big(j-\sqrt{j+1}\,\eta\big)\nonumber\\
&=(j+1)\eta^2+\sqrt{j+1}\,\eta-j(j+1)\label{eq:lambda},
\end{align}
so that the two position parameters $n$ and
$N-n$ combine, to leading order, into the single variable $\eta^2$. A similar mechanism was applied for the contraction of the Hahn oscillator of \cite{JSV}.
The dual
Hahn polynomials $R_t(\lambda(n);\alpha,\beta,K)$ obey the three term recurrence relation \cite[\S9.6]{KLS}
\begin{align}
\label{eq:dHrec}
&\lambda\,R_t=A_t\,R_{t+1}-(A_t+C_t)\,R_t+C_t\,R_{t-1},
\end{align}
where the coefficients are given by
\begin{align}
    &A_t=(t+\alpha+1)(t-K),\\
&C_t=t(t-\beta-K-1).
\end{align}
For the sets of parameters in \eqref{eq:dualHahn}, the coefficients read
\begin{align}
\label{eq:AC}
&A^e_t=(t-j-\tfrac\gamma2)(t-j),\quad C^e_t=t(t-\tfrac\gamma2)\,,\\
&A^o_t=(t-j+\tfrac\gamma2)(t-j),\quad C^o_t=t(t+\tfrac\gamma2).
\end{align}
Consider the even sub-lattice and set $\widehat R_t=(j+1)^tQ^e_n(t)$, so that $\widehat R_0=1$.
Substituting $R_t=(j+1)^{-t}\widehat R_t$ into \eqref{eq:dHrec} one gets
\begin{equation}
\label{eq:Rhatrec}
A^e_t\,\widehat R_{t+1}=(j+1)\big(\lambda+A^e_t+C^e_t\big)\,\widehat R_t-(j+1)^2C^e_t\,\widehat R_{t-1}\,.
\end{equation}
Every term here is $O(j^2)$.  Dividing by $(j+1)^2$ and using \eqref{eq:lambda}--\eqref{eq:AC}, each coefficient has a limit as $j\to\infty$
\begin{equation}
    \frac{A^e_t}{(j+1)^2}\to1,\qquad \frac{\lambda+A^e_t+C^e_t}{j+1}\to\eta^2-\big(2t+1-\tfrac\gamma2\big),
\qquad C^e_t=t\big(t-\tfrac\gamma2\big).
\end{equation}
Since the leading
coefficient tends to $1$, \eqref{eq:Rhatrec} determines $\widehat R_{t+1}$ from
$\widehat R_t$ and $\widehat R_{t-1}$, and induction on $t$ shows that each
$\widehat R_t$ converges, the limits satisfying
\begin{equation}
\label{eq:Lagrec}
\eta^2\widehat{R}_t=\widehat R_{t+1}+\big(2t+1-\tfrac\gamma2\big)\widehat R_t
+t\big(t-\tfrac\gamma2\big)\widehat R_{t-1},\qquad \widehat R_0=1,
\end{equation}
which is exactly the recurrence relation of the monic Laguerre polynomials \cite[\S9.12]{KLS} of parameter $-\gamma/2$ in the
variable $y=\eta^2$.
Thus
\begin{equation}
    (j+1)^tQ_n^e(t)\to(-1)^tt!\,L_t^{(-\gamma/2)}(\eta^2)=(-\gamma/2+1)_t (-1)^t{}_1F_1\!\left(\begin{matrix} -t\\-\frac{\gamma}{2}+1\end{matrix};\eta^2\right)\,.
\end{equation}
 The odd sub-lattice runs through the same steps: since $A^o_t,C^o_t$ in
\eqref{eq:AC} are obtained from $A^e_t,C^e_t$ by $\gamma\to-\gamma$, the limit is the same
recurrence with $-\tfrac\gamma2$ replaced by $\tfrac\gamma2$. Hence
\begin{equation}
\label{eq:conflu}
(j+1)^t\,Q^e_n(t)\ \longrightarrow\ (-1)^t\,t!\,L_t^{(-\gamma/2)}(\eta^2),\qquad
(j+1)^t\,Q^o_n(t)\ \longrightarrow\ (-1)^t\,t!\,L_t^{(\gamma/2)}(\eta^2).
\end{equation}

\paragraph{The envelope from the normalization; the ground state.}
By \eqref{eq:def-un} and \eqref{eq:Q}, the wavefunctions are given by 
\begin{equation}
\label{eq:ampl0}
u_{2t}(n)=\kappa_n\sqrt{\frac{w_{2t}}{h_n}}\,Q^e_n(t)
,\qquad
u_{2t+1}(n)=\kappa_n\sqrt{\frac{w_{2t+1}}{h_n}}\frac{(-j+\tfrac\gamma2)_n}{(-j-\tfrac\gamma2)_n}\,Q^o_n(t),\quad  n\le j\,,
\end{equation}
and since $Q^e_n(0)=Q^o_n(0)=1$
\begin{equation}
  u_0(n)=\sqrt{\frac{w_0}{h_n}}\,\kappa_n\,,\qquad u_1(n)=\sqrt{\frac{w_1}{h_n}}\frac{(-j+\tfrac\gamma2)_n}{(-j-\tfrac\gamma2)_n}\kappa_n \,.
\end{equation}
Dividing \eqref{eq:ampl0} by these two expressions removes $\kappa_n$ and $h_n$
altogether, leaving the excited state wavefunctions as simple multiples of the ground-state ones:
\begin{equation}
\label{eq:ampl}
u_{2t}(n)=\sqrt{\frac{w_{2t}}{w_0}}\;u_0(n)\,Q^e_n(t),\qquad
u_{2t+1}(n)=\sqrt{\frac{w_{2t+1}}{w_1}}\;u_1(n)\,Q^o_n(t).
\end{equation}
This is why the ground-state deserves separate treatment: once their limits are
known, the limits of the excited states follow from \eqref{eq:conflu} and the
asymptotics of the weight ratios $w_{2t}/w_0$ and $w_{2t+1}/w_1$, computed
below. To find the limit of $u_0$, we compare it with the ground state of the
Krawtchouk Hamiltonian, whose asymptotics are classical. We denote by
$\kappa_n^K$, $h_n^K$ and $w_0^K$ the values of $\kappa_n$, $h_n$ and
$w_0$ at $\gamma=1$
\begin{align}
 &\kappa^K_n=(-1)^n2^{-n}\,n!\binom Nn\,,\\
 &h^K_n=4^{-n}(n!)^2\binom Nn\,,\\
&w^K_0=2^{-N}
\label{krawtchouk_weight},
 \end{align}
so that the Krawtchouk ground state is
\begin{equation}
\label{eq:krgs}
u^K_0(n)=(-1)^n\sqrt{2^{-N}\tbinom Nn}\,.
\end{equation}
 Comparing $u_0$ with $u_0^K$ isolates the $\gamma$-dependence, which enters only
through the factor $(-j-\tfrac\gamma2)_n$ in $\kappa_n$ and through the ratio
$h_n/h^K_n=\prod_{l=1}^{n}U_l/U^K_l$, given by \eqref{eq:impurity}. For $n\le j$,
\begin{equation}
\label{eq:envelope}
\frac{u_0(n)}{u_0^{K}(n)}=\sqrt{\frac{w_0}{w_0^K}}\;
\frac{(-j-\tfrac\gamma2)_n}{(-j-\tfrac12)_n}\,
\Bigg[\prod_{l=1}^{n}\frac{4(j+1-l)^2-1}{4(j+1-l)^2-\gamma^2}\Bigg]^{1/2}.
\end{equation}
The asymptotics of \eqref{eq:envelope} are handled by 
\begin{equation}
\label{asymptotics-gamma}
\frac{\Gamma(z+a)}{\Gamma(z+b)}\sim z^{a-b},\qquad z\to\infty,
\end{equation}
which follows from Stirling's formula. The Pochhammer ratio can be written as
\begin{equation}
\label{eq:pochratio}
\frac{(-j-\tfrac\gamma2)_n}{(-j-\tfrac12)_n}
=\frac{\Gamma(j+1+\tfrac\gamma2)}{\Gamma(j+\tfrac32)}\;
\frac{\Gamma(j+1-n+\tfrac12)}{\Gamma(j+1-n+\tfrac\gamma2)}.
\end{equation}
Remembering $n=j+1+\eta\sqrt{j+1}$ and $n\le j$, we have $\eta<0$, and $j+1-n=   |\eta|\sqrt{j+1}>0$ for the distance of the site $n$ from the middle. The asymptotics \eqref{asymptotics-gamma} gives
\begin{align}
\frac{\Gamma(|\eta|\sqrt{j+1}+\tfrac12)}{\Gamma(|\eta|\sqrt{j+1}+\tfrac\gamma2)}&\sim|\eta|^{\frac{1-\gamma}2}{j}^{\frac{1-\gamma}{4}},\\
  \frac{\Gamma(j+1+\frac{\gamma}{2})}{\Gamma(j+\frac{3}{2})}&\sim j^{\frac{\gamma-1}{2}}.
\end{align}
The weight associated to the first spectral point is
\begin{equation}
    w_0=\frac{2^{-N}(1-\gamma/2)_j}{(1/2)_j},
\end{equation}
and by  \eqref{asymptotics-gamma} we have
\begin{equation}
  \sqrt{\frac{w_0}{w_0^K}}=\sqrt{\frac{(1-\gamma/2)_j}{(1/2)_j}}=\sqrt{\frac{\Gamma(j+1-\gamma/2)\Gamma(1/2)}{\Gamma(j+1/2)\Gamma(1-\gamma/2)}}\sim \Big(\frac{\sqrt\pi}{\Gamma(1-\gamma/2)}\Big)^{1/2} j^{\frac{1-\gamma}{4}}
\end{equation}
Each factor in the product of \eqref{eq:envelope} can be rewritten as
\begin{equation}
\frac{4(j+1-l)^2-1}{4(j+1-l)^2-\gamma^2}
=1+\frac{\gamma^2-1}{4(j+1-l)^2-\gamma^2},
\qquad l=1,\dots,n.
\end{equation}
Taking logarithms and using $\ln(1+x)\sim x$ for $x\sim 0$, the logarithm of the whole
product is bounded by a constant times $\sum_{l=1}^{n}(j+1-l)^{-2}=\sum_{m=j+1-n}^{j}m^{-2}\le
\tfrac1{j-n}$, and $j+1-n=|\eta|\sqrt{j+1}$, so for any fixed $\eta\neq0$ the whole product
tends to $1$ as $j\to\infty$. Combining these estimates, the r.h.s.\ of \eqref{eq:envelope} behaves as
\begin{equation}
\frac{u_0(n)}{u_0^{K}(n)}\sim \Big(\frac{\sqrt\pi}{\Gamma(1-\gamma/2)}\Big)^{1/2}|\eta|^{\frac{1-\gamma}{2}}\,.
\end{equation}
The classical piece is the Krawtchouk ground state itself: for $n=j+1+\sqrt{j+1}\,\eta$,
the De Moivre--Laplace theorem gives the local Gaussian form of the binomial weight,
\begin{equation}
2^{-N}\binom Nn\sim\big(\pi(j+1)\big)^{-1/2}e^{-\eta^2},
\end{equation}
so that as $j\to\infty$ we have
\begin{equation}
    (j+1)^{1/4}|u_0^K(n)|\sim \pi^{-1/4}e^{-\eta^2/2}\,. 
\end{equation}
 Putting the pieces together we conclude
\begin{equation}
\label{eq:gs}
(-1)^n\,(j+1)^{1/4}\,u_0(n) \sim\ \Gamma(1-\tfrac\gamma2)^{-1/2}\,|\eta|^{\frac{1-\gamma}{2}}\,e^{-\eta^2/2},
\end{equation}
the ground state of \eqref{eq:singosc}. For $n>j$, the mirror relation gives $u_0(N-n)=\epsilon_0u_0(n)=-u_0(n)$, and
$(-1)^{N-n}=-(-1)^n$, so the envelope $(-1)^nu_0(n)$ takes the same value at $n$ and $N-n$.
Since $n\mapsto N-n$ is $\eta\mapsto-\eta-1/\sqrt{j+1}$, the limit is even in $\eta$.
The ground state of the odd sub-lattice is treated in the same way. By
\eqref{eq:ampl0}, 
\begin{equation}
u_1(n)=\sqrt{\frac{w_1}{w_0}}\;
\frac{(-j+\tfrac\gamma2)_n}{(-j-\tfrac\gamma2)_n}\;u_0(n),
\end{equation}
where the Pochhammer ratio behaves as
\begin{equation}
    \frac{\Gamma(j+1-\tfrac\gamma2)\Gamma(j+1-n+\tfrac\gamma2)}{\Gamma(j+1+\tfrac\gamma2)\Gamma(j+1-n-\tfrac\gamma2)}
\sim|\eta|^{\gamma}j^{-\gamma/2}\,,
\end{equation}
 which raises the exponent by $\gamma$, while
 \begin{equation}
\sqrt{\frac{w_1}{w_0}}=\sqrt{\frac{(1+\gamma/2)_j}{(1-\gamma/2)_j}}\sim\;\left(\frac{\Gamma(1-\tfrac\gamma2)}{\Gamma(1+\tfrac\gamma2)}\right)^{1/2}j^{\gamma/2}
 \end{equation}
supplies the compensating power of $j$. We get
\begin{equation}
\label{eq:gsodd}
(-1)^n(j+1)^{1/4}\,u_1(n)\sim \Gamma(1+\tfrac\gamma2)^{-1/2}\,|\eta|^{\frac{1+\gamma}{2}}\,e^{-\eta^2/2}.
\end{equation}
For $n>j$, the mirror relation gives $u_1(N-n)=\epsilon_1u_1(n)=u_1(n)$, so the envelope
$(-1)^nu_1(n)$ changes sign between $n$ and $N-n$, and its limit is odd in $\eta$.
\begin{remark}
The two $\gamma$-dependent powers of $|\eta|$ come entirely from the Pochhammer
factors in \eqref{eq:Q}. Note that
$\tfrac{1\mp\gamma}2$ are the two roots of $\sigma(\sigma-1)=\tfrac{\gamma^2-1}4$, i.e.\ the
two Frobenius exponents of \eqref{eq:singosc} at $\eta=0$. The bulk equation of \S\ref{sec:eqlimit} fixes these two exponents but not which sub-lattice takes which; this information is
lost in the central layer and is supplied by the polynomials.
\end{remark}

\paragraph{The excited states.}
It remains to determine the asymptotics of the weight ratios
$w_{2t}/w_0$ and $w_{2t+1}/w_1$.  By \eqref{weight_even} and \eqref{weight_odd},
for fixed $t$ each of the Pochhammer symbols $(-j)_t$, $(-\tfrac\gamma2-j)_t$ and
$(\tfrac\gamma2-j)_t$ is a product of $t$ factors linear in $j$, so that
\begin{equation}
\label{eq:wratio}
\sqrt{\frac{w_{2t}}{w_0}}\sim\frac{j^{t}}{\sqrt{t!\,(1-\gamma/2)_t}}\,,\qquad
\sqrt{\frac{w_{2t+1}}{w_1}}\sim\frac{j^{t}}{\sqrt{t!\,(1+\gamma/2)_t}}\,.
\end{equation}
Substituting \eqref{eq:wratio} into \eqref{eq:ampl}, the power $j^t$ cancels  against
the $(j+1)^{-t}$ of the confluence \eqref{eq:conflu}.
\begin{proposition}[Confluence to the generalized Hermite functions]
\label{prop:conv}
Let $\gamma\in(0,2)$ and $t\in\{0,1,2,\dots\}$ be fixed, let $N=2j+1\to\infty$ and write $n=j+1+\eta\sqrt{j+1}$. Then \eqref{eq:conflu} holds, and, if $\eta\neq0$,
the wavefunctions \eqref{eq:def-un} satisfy
\begin{equation}
\label{eq:genH}
(-1)^{n+t}\,(j+1)^{1/4}\,u_{2t}(n)\ \to\ \psi_{2t}(\eta),\qquad
(-1)^{n+t+1}\,(j+1)^{1/4}\,u_{2t+1}(n)\ \to\ \psi_{2t+1}(\eta),
\end{equation}
where
\begin{equation}
\label{eq:cNt}
\begin{aligned}
\psi_{2t}(\eta)&=\sqrt{\frac{t!}{\Gamma(t+1-\gamma/2)}}\;
|\eta|^{\frac{1-\gamma}2}e^{-\eta^2/2}L_t^{(-\gamma/2)}(\eta^2),\\
\psi_{2t+1}(\eta)&=\sqrt{\frac{t!}{\Gamma(t+1+\gamma/2)}}\;
\operatorname{sgn}(\eta)\,|\eta|^{\frac{1+\gamma}2}e^{-\eta^2/2}L_t^{(\gamma/2)}(\eta^2),
\end{aligned}
\end{equation}
and the functions $\psi_s$ are normalized in $L^2(\mathbb R,d\eta)$.
\end{proposition}
\begin{proof}
The proof follows from the previous discussion. For $\eta<0$, i.e.\ $n\le j$, the limits follow by combining \eqref{eq:ampl},
\eqref{eq:conflu}, \eqref{eq:wratio}, \eqref{eq:gs} and \eqref{eq:gsodd}. Combining \eqref{eq:genH}, valid for $n\le j$, with the mirror relation \eqref{eq:mirror} extends the limit to
$n>j$, i.e.\ $\eta>0$, making the even--sub-lattice envelopes even and the odd--sub-lattice ones odd in $\eta$.
This gives the limit of every eigenvector. The normalization of  $\psi_s$ follows from
\begin{equation}
    \int_{\mathbb R}|\eta|^{1\mp\gamma}e^{-\eta^2}L_t^{(\mp\gamma/2)}(\eta^2)^2\,d\eta
=\int_0^\infty y^{\mp\gamma/2}e^{-y}L_t^{(\mp\gamma/2)}(y)^2\,dy=\Gamma(t+1\mp\tfrac\gamma2)/t!\,, 
\end{equation}
and $\|\psi_s\|^2=1$.
\end{proof}
\begin{remark}
The right-hand sides of \eqref{eq:genH} are the even and odd halves of two families of
generalized Hermite functions \cite{Rosenblum} (\S\ref{sec:reflection}). A direct
substitution shows that they are exact eigenfunctions of the same singular oscillator
\eqref{eq:singosc}, with eigenvalues $2x_{2t}+2-\gamma=4t+2-\gamma$ and
$2x_{2t+1}+2-\gamma=4t+2+\gamma$, respectively built on the two Frobenius exponents
$\frac{1\mp\gamma}{2}$ at the origin. At $\gamma=1$ (exponents $0,1$), these are the even and odd
ordinary Hermite functions.
\end{remark}

\subsection{The singular point: self--adjoint extension and the monopole analogy}
\label{sec:extension}
Consider $\Hs$ in \eqref{eq:Hdef} as an operator defined on smooth functions vanishing in a
neighbourhood of the origin. Near the singular point $\eta=0$, the solutions behave as
$|\eta|^{\sigma_\mp}$ with $\sigma_\mp=\frac{1\mp\gamma}{2}$, and both are square integrable
since $0<\gamma<2$, so that $\eta=0$ is in the limit-circle case. Hence $\Hs$ is not
essentially self-adjoint, and a boundary condition at $\eta=0$ must be imposed.

This is the situation analysed in
\cite{DHV} for the s-wave sector of a spin-$\tfrac12$ particle in the field of a magnetic
monopole, where the monopole coupling cancels the centrifugal barrier. There the radial
operator on the half-line $r>0$, and its counterpart $-\tfrac{d^2}{dr^2}+\tfrac{g}{r^2}$
when a $1/r^2$ potential is added, have a one-parameter family of self-adjoint extensions.
 The dilation generator of the $\so$ dynamical algebra is compatible with the boundary
condition of exactly two of them, those on which the eigenfunctions are pure Frobenius
solutions, called the symmetric and antisymmetric extensions in \cite{DHV}. For the compact
generator $-\tfrac{d^2}{dr^2}+\tfrac g{r^2}+r^2$ of the same algebra, these two extensions
carry the eigenfunctions $r^{\sigma}e^{-r^2/2}L_n^{(\sigma-1/2)}(r^2)$, with $\sigma$ one of
the two exponents (the wave functions that \cite{DHV} obtains by algebraic methods).
It is only on them that $\so\cong\su$ is realized with self--adjoint generators
\cite[\S V]{DHV}.

The functions $\psi_{2t}$ and $\psi_{2t+1}$ of Proposition~\ref{prop:conv}
correspond to these two distinguished extensions, the even functions
realizing the one with exponent $\sigma_-$ and the odd functions the one
with exponent $\sigma_+$, so that the two parts of $L^2(\mathbb R)$ carry
different extensions. This choice is not determined by the bulk equation \eqref{eq:singosc}; it is fixed by the
discrete model, which selects the extension of each sub-lattice through the polynomials. Since the eigenvalue for each parity is $4t+2\sigma+1$, this choice is also
what produces the offset $2\gamma$ between the two families of eigenvalues of $\Hs$. By abuse of notation, $\Hs$ denotes both the differential expression \eqref{eq:Hdef} and its self-adjoint
extension selected by the discrete model, and which one is meant will
be clear from the context.

\section{The reflection and the two-channel structure}
\label{sec:reflection}

Throughout this section $\Hs$ is the self-adjoint extension of \eqref{eq:Hdef} selected by
the discrete model (\S\ref{sec:extension}). By \eqref{eq:cNt} the functions $\psi_{2t}$ are
even and the $\psi_{2t+1}$ are odd. The substitution $y=\eta^2$ identifies the even and odd
parts of $L^2(\mathbb R)$ with $L^2\big((0,\infty),y^{\mp\gamma/2}e^{-y}dy\big)$, and maps
$\psi_{2t}$ and $\psi_{2t+1}$ to the Laguerre polynomials $L_t^{(\mp\gamma/2)}$, which are
complete since $\mp\gamma/2>-1$ \cite{Akhiezer1965}. Hence $\{\psi_{2t}\}$ and $\{\psi_{2t+1}\}$ are
orthonormal bases of the even and odd subspaces, and together an orthonormal basis of
$L^2(\mathbb R)$. In this basis $\Hs$ is diagonal, with eigenvalues $4t+2-\gamma$ and
$4t+2+\gamma$, and since every $\psi_s$ has definite parity it commutes with the reflection
\begin{equation}
\label{eq:refl}
\R:\ \eta\mapsto-\eta,\qquad [\R,\Hs]=0.
\end{equation}
We call the two eigenspaces of $\R$ the two channels of $\Hs$, and the two sequences of
eigenvalues its two towers.

The coefficient $\frac{\gamma^2-1}{4}$ of the $\sfrac{1}{\eta^2}$ term is the same in both channels, so
the reflection does not enter the potential, and the channels differ only in the boundary
condition at the origin (\S\ref{sec:extension}). The two towers are separated by $2\gamma$.
For the operator $\N=\frac12(\Hs-2+\gamma)$, the continuum image of $J$ (\S\ref{sec:FR}),
the separation is $\gamma$ and the spectrum is $\{2t\}_{t\ge0}\cup\{2t+\gamma\}_{t\ge0}$.

\paragraph{Discrete origin of the reflection.} The reflection comes from the site reversal $\Rc:\ n\mapsto N-n$. Since the
lattice has an even number $N+1=2j+2$ of sites, $\Rc$ is centred on the middle
point, and in the limit it acts on the continuum variable as $\eta\mapsto-\eta$.
It also flips the staggering sign, $(-1)^{N-n}=-(-1)^n$, so that on the
envelopes \eqref{eq:envelopePhi} the site reversal acts as minus the reflection,
\begin{equation}
\label{eq:RcR}
\Rc=-\R\quad\text{on the envelopes.}
\end{equation}
This agrees with \eqref{eq:eps}; since $\epsilon_s=(-1)^{s+1}$ for $N=2j+1$, the eigenvectors of the
even sub-lattice ($\epsilon=-1$) have even envelopes and those of the odd sub-lattice
($\epsilon=+1$) have odd envelopes, as in \eqref{eq:cNt}.

\paragraph{Why this is not the Wigner--Dunkl oscillator.}
It is tempting to write the $\sfrac{1}{\eta^2}$ term as the Dunkl potential $\mu(\mu-\R)/\eta^2$ and
call $\Hs$ parabosonic. This is wrong for $\gamma\neq1$. That potential has
parity dependent coefficients $\mu(\mu\mp1)$, whose Frobenius exponents differ by
exactly one. Matching the even exponent $\tfrac{1-\gamma}{2}$ fixes $\mu=\tfrac{1-\gamma}{2}$
but then predicts an odd exponent $\mu+1=\tfrac{3-\gamma}{2}$, whereas the model has
$\tfrac{1+\gamma}{2}$; they agree only at $\gamma=1$. Equivalently, no potential
$(\mu_0+\mu_1\R)/\eta^2$ reproduces two channels of the same coefficient unless
$\mu_1=0$. Thus
$\Hs$ is the scalar singular oscillator, with $\R$ grading the boundary condition
selected by the discrete model and absent from the potential; it is not a Wigner--Dunkl
oscillator away from $\gamma=1$.

\paragraph{Two complementary generalized--Hermite families.}
The generalized Hermite functions of parameter $\mu$ \cite{Rosenblum} are the eigenfunctions
of the Wigner--Dunkl oscillator $-\partial_\eta^2+\eta^2+\mu(\mu-\R)/\eta^2$: the even ones,
$|\eta|^{\mu}e^{-\eta^2/2}L_t^{(\mu-1/2)}(\eta^2)$, and the odd ones,
$\operatorname{sgn}(\eta)|\eta|^{\mu+1}e^{-\eta^2/2}L_t^{(\mu+1/2)}(\eta^2)$. The even
channel of \eqref{eq:genH} is the even family of parameter $\mu=\tfrac{1-\gamma}{2}$
(Laguerre index $\mu-\tfrac12=-\tfrac\gamma2$); the odd channel is the odd family of the
opposite parameter $-\mu=\tfrac{\gamma-1}2$ (exponent $-\mu+1=\tfrac{1+\gamma}2$, Laguerre
index $-\mu+\tfrac12=\tfrac\gamma2$). The para-Krawtchouk
continuum is thus the singular oscillator whose two channels are the even part of the
parabosonic family $+\mu$ and the odd part of the family $-\mu$, a $\mu\leftrightarrow-\mu$
combination that degenerates to the single ordinary oscillator at $\gamma=1$.

\section{Algebraic structure: from the Hahn algebra to $\su$}
\label{sec:algebra}
In the finite model \cite{VZ,BCLMNV}, the multiplication operator $X=J$, which is tridiagonal in the basis $\ket{n}$, and the difference operator $Y$, which is diagonal in that basis, generate the
quadratic Hahn algebra $\halg$ (gothic letters denote the algebras named after the polynomial
families, the subscript their rank) whose structure constants depend on $N$ and
$\gamma$ \cite{VZ}. The
$(N+1)$--dimensional module splits along the two sublattices into two $\halg$--submodules,
the $\Rc=\mp1$ sectors, which become the two channels of the limit \cite{BCLMNV}. By \eqref{eq:singosc} and \eqref{eq:lambda}, the two generators have the continuum limits
\begin{equation}
\label{eq:XYlimit}
X=J\ \longrightarrow\ \tfrac12(\Hs+\gamma-2),\qquad
\frac{Y+j(j+1)}{j+1}\ \longrightarrow\ \eta^2 ,
\end{equation}
that is, a shift of the singular oscillator Hamiltonian and a shifted and rescaled
multiplication by $\eta^2$.
Together with the dilation $D=1+2\eta\partial_\eta$ they close on the dynamical algebra of the
singular oscillator,
\begin{equation}
\label{eq:su11}
[\Hs,\eta^2]=-2D,\qquad [\Hs,D]=4\Hs-8\eta^2,\qquad [\eta^2,D]=-4\eta^2,
\end{equation}
which is $\su\cong\mathfrak{sl}(2,\mathbb R)\cong\so$, with compact generator
$K_0=\tfrac14\Hs$. Thus the two generators of the quadratic Hahn algebra become, in the
limit, two generators of the Lie algebra $\su$: the
linear grid lets the position $\eta^2$ close linearly on $\Hs$ through the
dilation \eqref{eq:su11}. We do not contract the Hahn relations themselves,
since their structure constants depend on $N$. The statement concerns only what the two
bispectral operators become in the limit.
A quadratic grid does
not lead to a Lie algebra: for the para-Racah Hamiltonian, whose bispectral pair generates
the Racah algebra $\ralg$, the limiting pair generates the quadratic Jacobi algebra $\jalg$
\cite{PRcont}.

Each channel of $\Hs$ carries a positive discrete series 
$\mathcal D^+_{\nu}$ of $\su$, with $K_0=\tfrac14\Hs$ having eigenvalues $\nu+t$, $t=0,1,2,\dots$, and by \eqref{eq:genH} the
two Bargmann indices are
\begin{equation}
\nu_{\mathrm{even}}=\frac{2-\gamma}{4},\qquad \nu_{\mathrm{odd}}=\frac{2+\gamma}{4}.
\end{equation} 
They differ by $\tfrac\gamma2$, which agrees with the separation $2\gamma$ of the two
towers of $\Hs$ and the reflection $\R$ acts as $+1$ on the first series and as $-1$ on the
second. At $\gamma=1$ the indices are $\tfrac14$ and $\tfrac34$, those of the even and odd
sectors of the ordinary oscillator, which $\R$ combines with $\su$ into the Lie superalgebra $\osp$
of the parabosonic oscillator with $\mu=0$. This is the algebraic image of
the analytic statement of \S\ref{sec:reflection}: the discrete Hahn bispectrality becomes the
$\su$ dynamical symmetry of the singular oscillator, $\mathbb Z_2$--graded by the mirror.

\section{The case of even $N$}
\label{sec:even}

We now take $p=0$, so that $N=2j$ and the chain has an odd number $N+1=2j+1$ of sites, with
a central site $n=j$. The bi-lattice \eqref{eq:bilattice} becomes
\begin{equation}
\label{eq:bilatticeeven}
x_{2s}=2s,\quad s=0,\dots,j,\qquad x_{2s+1}=2s+\gamma,\quad s=0,\dots,j-1.
\end{equation}
The computations are those of \S\S\ref{sec:eqlimit}--\ref{sec:polylimit}, and we only
indicate what changes.  This section shows that the continuum limit is the same as for $p=1$, and how the
parity of $N$ is absorbed.  As for the odd case, we work with $j$ directly and with a shifted and re-scaled
position, now measured from the central site,
\begin{equation}
    \label{eq:etaeven}
    \eta:=\frac{n-j}{\sqrt{j+\frac{1}{2}}}\,.
\end{equation}
with the same scale $\sqrt{(N+1)/2}$ as in \eqref{eq:etanj}. The mirror map $n\mapsto N-n$ is
now exactly $\eta\mapsto-\eta$. 
For $p=0$ the coefficients
\eqref{eq:recurrenceB}--\eqref{eq:recurrenceU} yield exactly
\begin{equation}
\label{eq:receven}
\begin{aligned}
\widetilde B_n&=\noleft(j+\frac{\noleft(\gamma-1\noright)}{2}\noright)+\frac{(\gamma-1)}{4\eta^2}\frac{1}{\noleft(1-\frac{1}{4\eta^2\noleft(j+\frac{1}{2}\noright)}\noright)},\\
\widetilde U_n&=\frac{\noleft(j+\frac{1}{2}\noright)^2}{4}\noleft(1-\noleft(\frac{\eta}{\sqrt{j+\frac{1}{2}}}-\frac{1}{2\noleft(j+\frac{1}{2}\noright)}\noright)^{\!\!\!2}\,\noright)\noleft(1-\noleft(\frac{\gamma-1}{2\sqrt{j+\frac{1}{2}}\eta-1}\noright)^{\!\!\!2}\,\noright).
\end{aligned}
\end{equation}

\paragraph{The impurity is shared between couplings and fields.}
We can note that the first two factors of $\widetilde U_n$ in \eqref{eq:receven} give exactly the Krawtchouk coupling
$U_n^K=\tfrac n4(2j+1-n)$ with $N=2j$, while the third gives the impurity of order $\sfrac{1}{\eta^2}$. However, when comparing with
\eqref{eq:impurity}, this impurity now has a different shape. The coupling factor depends on $\gamma$ only through $(\gamma-1)^2$ and is at most $1$. Thus, at every point, $U^K_n$ exceeds $\widetilde U_n$ by
$\sfrac{\noleft(\gamma-1\noright)^2}{\noleft(2\sqrt{j+\frac{1}{2}}\eta-1\noright)^2}$ with no sign reversal in the middle. Moreover, unlike $B_n$, $\widetilde B_n$ is not constant and also carries an impurity of order $\sfrac{1}{\eta^2}$. As such, away from the
centre, $\widetilde B_n$ exceeds its asymptotic value $\noleft(j+\frac{\noleft(\gamma-1\noright)}{2}\noright)$ by $\frac{(\gamma-1)}{4\eta^2}{\noleft(1-\sfrac{1}{4\eta^2\noleft(j+\frac{1}{2}\noright)}\noright)^{\!\!-1}}$ (positively for $\gamma>1$ and negatively for $\gamma<1$). At the central site the excess is $-(\gamma-1)(j+\frac{1}{2})$. This is the same reversal between tail and core
that the $U_n$ display for odd $N$. All of this is shown in Figure~\ref{fig:even}. Additionally, this is the counterpart, on
the linear grid, of an observation made on the para-Racah Hamiltonian \cite{paraRacahChain}: for
even $N$, both the terms associated with the couplings ($\widetilde U_n$) and the magnetic fields ($\widetilde B_n$) carry the central defect.

\begin{figure}[ht]
\centering
\begin{tikzpicture}
\begin{axis}[
  width=0.47\textwidth,height=6.2cm,
  xlabel={position $n$},ylabel={$\widetilde U_n/U_n^{K}$},
  title={\small(a) $\widetilde U_n$, $N=40$},
  xmin=1,xmax=40,ymin=0.6,ymax=1.05,
  legend style={at={(0.97,0.05)},anchor=south east,draw=none,font=\footnotesize},
  tick label style={font=\footnotesize},label style={font=\small},
]
  \addplot[royalblue,thick,mark=*,mark size=0.9pt,samples=40,domain=1:40]
    {1-0.25/((2*x-41)^2)};
  \addlegendentry{$\gamma=0.5$ and $\gamma=1.5$}
  \addplot[black,dashed,thick,samples=40,domain=1:40]{1};
  \addlegendentry{$\gamma=1$}
\end{axis}
\end{tikzpicture}\hfill
\begin{tikzpicture}
\begin{axis}[
  width=0.47\textwidth,height=6.2cm,
  xlabel={position $n$},ylabel={$\widetilde B_n-\widetilde B_{\infty}$},
  title={\small(b) $\widetilde B_n$, $N=40$},
  xmin=0,xmax=40,ymin=-11,ymax=11,
  legend style={at={(0.97,0.97)},anchor=north east,draw=none,font=\footnotesize},
  tick label style={font=\footnotesize},label style={font=\small},
]
  \addplot[royalblue,thick,mark=*,mark size=0.9pt,samples=41,domain=0:40]
    {(-0.5)*41/(2*((2*x-40)^2-1))};
  \addlegendentry{$\gamma=0.5$}
  \addplot[black,dashed,thick,samples=41,domain=0:40]{0};
  \addlegendentry{$\gamma=1$}
  \addplot[coralred,thick,mark=triangle*,mark size=1.1pt,samples=41,domain=0:40]
    {(0.5)*41/(2*((2*x-40)^2-1))};
  \addlegendentry{$\gamma=1.5$}
\end{axis}
\end{tikzpicture}
\caption{The impurity for even $N$, $N=40$. (a) The coupling ratio
$\widetilde U_n/U_n^K=1-\sfrac{\noleft(\gamma-1\noright)^2}{\noleft(2\sqrt{j+\frac{1}{2}}\eta-1\noright)^2}$: the coupling is weakened at every point, by an
amount that depends on $\gamma$ only through $(\gamma-1)^2$, so that $\gamma=\tfrac12$ and
$\gamma=\tfrac32$ give the same couplings. (b) The excess of $\widetilde B_n$ over its asymptotic
value $\widetilde B_{\infty}=\tfrac{N-1+\gamma}2$, namely $(\gamma-1)(j+\frac12)/\big(4(j+\frac12)\eta^2-1\big)$, which carries the sign of $\gamma-1$
and is reversed at the central site. In the continuum limit only the combination of the two
survives, and it is the same $\sfrac{1}{\eta^2}$ term as for $N$ odd.}
\label{fig:even}
\end{figure}
\paragraph{The same limit.}
In the continuum limit, these two contributions combine to give the same centrifugal term as for $N$ odd. Indeed, in the bulk, at a site of
fixed position $\eta\neq0$ with $|\eta|\gg\frac{1}{\sqrt{j+\frac{1}{2}}}$, we can similarly expand $\widetilde B_n$ and $\widetilde U_n$ to get:
\begin{equation}
\label{eq:recevenExpanded}
\begin{aligned}
\widetilde B_n&=j+{\frac{\noleft(\gamma-1\noright)}{2}+\frac{\noleft(\gamma-1\noright)}{4\eta^2}}+O\noleft(\frac{1}{j}\noright)\\
\sqrt{\widetilde U_n}&=\frac{(j+\frac{1}{2})}{2}-\frac{\noleft(\gamma-1\noright)^2}{16\eta^2}-\frac{\eta^2}{4}+O\noleft(\frac{1}{j}\noright),
\end{aligned}
\end{equation}
and using the staggered ansatz $p_n=(-1)^n\varphi(\eta)$ in the recurrence relation, we find that the terms of order $j$ cancel as
before, and the singular terms combine as
\begin{equation}
\label{eq:combine}
\frac{\gamma-1}{4\eta^2}+\frac{(\gamma-1)^2}{8\eta^2}=\frac{(\gamma-1)(\gamma+1)}{8\eta^2}
=\frac{\gamma^2-1}{8\eta^2},
\end{equation}
exactly the impurity term of the case of odd $N$. Hence, as in \S\ref{sec:eqlimit}, we recover the same singular oscillator
\eqref{eq:singosc} with the same eigenvalue $2x+2-\gamma$. 

We can thus note that the important difference between $N$ even and odd is that, at the discrete level, the marginal impurity is carried only by the $U_n$ terms for the odd case, while it is distributed amongst the $\widetilde U_n$ and $\widetilde B_n$ terms which then combine to give an identical bulk limit for the even case. Although the bulk continuum limit is the same, it is important to mention that at the central site $\eta=0$, these differences in the discrete regime do lead (by similar arguments as in \S\ref{sec:eqlimit}) to a different matching condition:
\begin{equation}
\label{eq:centreeven}
\varphi(0)=\frac12\sqrt{\frac{\gamma}{2-\gamma}}\;
\noleft(\varphi\bigg(\tfrac{-1}{\sqrt{j+\frac{1}{2}}}\bigg)+\varphi\bigg(\tfrac{1}{\sqrt{j+\frac{1}{2}}}\bigg)\noright),
\end{equation}
to leading order (with a smooth continuation only at $\gamma=1$). Being symmetric in the two
neighbours $\pm\noleft(j+\frac{1}{2}\noright)^{-\frac{1}{2}}$, \eqref{eq:centreeven} is satisfied identically by an $\R$--odd envelope and
constrains only the $\R$--even one. This is the counterpart of \eqref{eq:A:centre}, and, as
for $N$ odd, the boundary behaviour at the singular point can be directly read off the polynomials.

\paragraph{The polynomials and the mirror.}
For $N=2j$ and $n\le j$ the polynomials on the bi-lattice are \cite{BCLMNV}
\begin{equation}
\label{eq:Qeven}
\begin{aligned}
P_n(x_{2t})&=\kappa_n\,R_t\big(\lambda(n);\alpha,\beta,j\big),\\
P_n(x_{2t+1})&=\kappa_n\,\frac{(j-n)\,(-j+\tfrac\gamma2)_n}{j\,(-j+1-\tfrac\gamma2)_n}\;
R_t\big(\lambda(n);\beta,\alpha,j-1\big),
\end{aligned}
\end{equation}
with $\kappa_n$ as in \eqref{eq:P3F2} for $p=0$, $\lambda(n)=n(n-N)$ and
$(\alpha,\beta)=(-j-\tfrac\gamma2,\,-j-1+\tfrac\gamma2)$.
\begin{remark}
At $n=j$ the factor $(j-n)$ vanishes: $P_j$ vanishes on the
whole odd sub-lattice, whose $j$ points are its zeros. This is forced by the persymmetry
with a central site: \eqref{eq:mirror} at $n=j=N-j$ reads $p_j(x_s)=\epsilon_sp_j(x_s)$,
so that $p_j(x_s)=0$ whenever $\epsilon_s=-1$, and for $N$ even $\epsilon_s=(-1)^{N+s}=(-1)^s$
is $-1$ exactly on the odd sub-lattice.
\end{remark}
With $\eta$ as in \eqref{eq:etaeven}, 
\begin{equation}
    \lambda(n)=(j+\tfrac12)\eta^2-j^2\,,
\end{equation} now without a linear term. For the parameters $  (\alpha,\beta)=(-j-\tfrac\gamma2,\,-j-1+\tfrac\gamma2)$, the dual Hahn polynomials obey the recurrence  \eqref{eq:dHrec} with the coefficients 
\begin{align}
    &C^e_t=t(t-\tfrac\gamma2)\,,\qquad A^e_t=(t-j+1-\tfrac\gamma2)(t-j)\,,\\
    &C^o_t=t(t+\tfrac\gamma2)\,,\qquad A^o_t=(t-j+\tfrac\gamma2)(t-j+1)\,.
\end{align}
These
differ from \eqref{eq:AC} only at lower order, and the argument of \S\ref{sec:polylimit} gives
the same Laguerre limits \eqref{eq:conflu}.
The envelopes follow as in \S\ref{sec:polylimit}: for $n\le j$, with $j-n=|\eta|\sqrt{j+\tfrac12}$,
the $\gamma$-dependent factors behave as
\begin{equation}
\label{even_approx}
\frac{(-j+1-\tfrac\gamma2)_n}{(-j+\tfrac12)_n}\sim |\eta|^{\frac{1-\gamma}{2}}j^{\frac{\gamma-1}{4}},\qquad
\frac{(j-n)\,(-j+\tfrac\gamma2)_n}{j\,(-j+1-\tfrac\gamma2)_n}\sim|\eta|^\gamma\,j^{-\gamma/2},
\end{equation}
and the powers of $j$ are again compensated by the weights. Proposition~\ref{prop:conv}
therefore holds for $N=2j$, with $\eta$ given by \eqref{eq:etaeven} and $(j+1)^{1/4}$ replaced
by $(j+\tfrac12)^{1/4}$.

The mirror
bookkeeping of \S\ref{sec:reflection} changes in two places at once. The site reversal
$\Rc$ is now centred on the central site, and leaves the
staggering $(-1)^n$ invariant, so that it acts on the envelopes as $+\R$; and
$\epsilon_s=(-1)^{N+s}=(-1)^s$ is now $+1$ on the even sub-lattice and $-1$ on the odd one.
The two changes compensate: the even sub-lattice is again $\R$--even, with the exponent
$\tfrac{1-\gamma}2$, and the odd sub-lattice $\R$--odd, with $\tfrac{1+\gamma}2$. The two
parities of $N$ thus produce one and the same limit, with one and the same grading, the two
sign changes compensating as they must.

\section{Relation with other discrete oscillators}
\label{sec:relations}
 
In this section we compare the para-Krawtchouk chain and its continuum limit with three
related constructions: the Krawtchouk and Hahn oscillators, the $\su$ Meixner oscillator on the half-lattice and the doubling procedure for orthogonal polynomials.

\subsection{Comparison with the Krawtchouk and Hahn oscillators}
\label{sec:compare}

It is useful to compare the para-Krawtchouk chain with the two finite oscillators of the
Askey scheme closest to it. The $\sutwo$ (Krawtchouk) finite oscillator contracts to the
canonical oscillator, with Hermite wavefunctions. The $\utwo_\alpha$ (Hahn) finite oscillator
of Jafarov, Stoilova and Van der Jeugt \cite{JSV} contracts to the parabosonic oscillator
($\osp$). Its wavefunctions tend to the generalized Hermite functions
$|\eta|^{\alpha+1/2}e^{-\eta^2/2}L_n^{(\alpha)}(\eta^2)$ (even) and
$\eta|\eta|^{\alpha+1/2}e^{-\eta^2/2}L_n^{(\alpha+1)}(\eta^2)$ (odd) of a single parameter,
whose exponents differ by one \cite[eqs.~(33), (35)]{JSV}. The para-Krawtchouk chain behaves
differently. Its couplings are the Krawtchouk ones with a central impurity, and its continuum
limit is the singular oscillator ($\su$), with even and odd channels given by
$L_t^{(-\gamma/2)}$ and $L_t^{(\gamma/2)}$. For $\gamma\neq1$ the two Laguerre indices do not
differ by one, and the exponents differ by $\gamma$.

\begin{center}
\renewcommand{\arraystretch}{1.25}
\begin{tabular}{@{}llll@{}}
\toprule
finite model & couplings & continuum limit & wavefunctions\\
\midrule
Krawtchouk ($\sutwo$) & uniform & harmonic osc.\ (osc.\ alg.) & $H_n(\eta)$\\
Hahn ($\utwo_\alpha$) \cite{JSV} & Hahn & parabose osc.\ ($\osp$) &
$L_n^{(\alpha)}$, $L_n^{(\alpha+1)}$\\
{para-Krawtchouk} & Krawtchouk $+$ impurity & {singular osc.\ ($\su$)} &
$L_t^{(-\gamma/2)}$, $L_t^{(+\gamma/2)}$\\
\bottomrule
\end{tabular}
\end{center}
The relation with the Hahn oscillator comes from the decomposition of \cite{BCLMNV} used in
\S\ref{sec:polylimit}. The two sub-lattices of the finite model, which become the two
channels of $\R$ in the limit, carry two dual Hahn families with the parameters $\alpha$ and
$\beta$ interchanged \eqref{eq:dualHahn}. Each sub-lattice alone contracts to one half, even
or odd, of a generalized Hermite family, while the Hahn oscillator of \cite{JSV} contracts to
a complete family.
 The para-Krawtchouk
Hamiltonian combines the two sub-lattices on one line, and because the two parabosonic
parameters are opposite ($\mu$ and $-\mu$ with $\mu=\tfrac{1-\gamma}2$, i.e.,\ Laguerre
indices $\mp\gamma/2$, \S\ref{sec:reflection}), the union is the scalar singular oscillator
rather than a single parabose oscillator.

\subsection{Comparison with the $\su$ Meixner oscillator on the half-lattice}
\label{sec:meixner}

The singular oscillator is also the continuum limit of a rather different discrete model, and
the comparison illuminates both. On $\ell^2(\mathbb N_0)$, with $m\in\mathbb N_0$ the
lattice variable, the positive discrete series
$\mathcal D^+_\nu$ of $\su$ is realized by difference operators, with compact generator
\begin{equation}
\label{eq:meixnerJ0}
J_0=\nu-\frac{1}{1-c}\Big[\hat m\,(S^{-1}-1)+c\,(\hat m+2\nu)(S-1)\Big],\qquad
J_0\,M_n(m;2\nu,c)=(n+\nu)\,M_n(m;2\nu,c),
\end{equation}
where $S^{\pm1}f(m)=f(m\pm1)$, $\hat m$ is multiplication by $m$, and
\begin{equation}
M_n(m;2\nu,c)={}_2F_1\left(\begin{matrix} -n,-m\\2\nu\end{matrix};1-1/c\right)
\end{equation}
are the Meixner
polynomials \cite[\S9.10]{KLS}. This is the Meixner oscillator of Atakishiyev, Jafarov,
Nagiyev and Wolf \cite{AJNW} (see also \cite{JVdJ}). The parameter $c\in(0,1)$ is pure
gauge in a precise sense: the family \eqref{eq:meixnerJ0} is the orbit of the diagonal
operator $\hat m+\nu$ under conjugation by a hyperbolic one-parameter subgroup
($c=\tanh^2(\tau/2)$) \cite{AJNW}, so the spectrum $\nu+\mathbb N_0$ and the Casimir
$\nu(\nu-1)$ are the same for every $c$. 

Setting $c=1-\Delta$, $m=y/\Delta$ and letting
$\Delta\to0^+$, the $O(\Delta^{-1})$ terms cancel between the two brackets of
\eqref{eq:meixnerJ0} and 
\begin{equation}
   J_0\to\nu-[y\partial_y^2+(2\nu-y)\partial_y]\,.
\end{equation}
This is the Laguerre
operator, whose eigenfunctions with eigenvalue $n+\nu$ are the Laguerre polynomials $L^{(2\nu-1)}_n(y)$. With
$y=\eta^2$ ($\eta>0$) and the gauge $\Phi=\eta^{\sigma}e^{-\eta^2/2}$, $\sigma=2\nu-\tfrac12$,
one finds
\begin{equation}
\label{eq:meixnerlimit}
\Phi\,J_0\,\Phi^{-1}
=\tfrac14\Big[-\partial_\eta^2+\eta^2+\frac{\sigma(\sigma-1)}{\eta^2}\Big],
\end{equation}
with eigenfunctions $\ \eta^{\sigma}e^{-\eta^2/2}L^{(2\nu-1)}_n(\eta^2)$. Thus,
the singular oscillator on the half--line, is reached through the classical
Meixner--Laguerre confluence.

The bridge with the para-Krawtchouk limit is exact. Setting $\sigma$ equal to the exponents $\sigma_\mp=\tfrac{1\mp\gamma}2$ of the two channels
(\S\ref{sec:reflection}) gives $\nu=\tfrac{2\mp\gamma}4$, which are the Bargmann indices
$\nu_{\rm even}$ and $\nu_{\rm odd}$ of \S\ref{sec:algebra}. The Laguerre indices
$2\nu-1=\mp\gamma/2$ agree with \eqref{eq:genH}, and $\Hs=4J_0$ in each channel. The para-Krawtchouk continuum thus consists of two Meixner oscillators, with complementary
Bargmann indices $\tfrac{2\mp\gamma}4$. Each $\R$-channel is the even or the odd extension to
the whole line of the limit of a Meixner oscillator on the half-line. The Meixner model, on
the other hand, describes a single channel. It has no mirror symmetry, no bi-lattice and no
fractional revival, since fractional revival comes from the pairing of the two channels,
that is, from the offset $\gamma$ between the two towers. The two roads to the same
operator are otherwise very different:

\begin{center}
\renewcommand{\arraystretch}{1.2}
\footnotesize
\begin{tabular}{@{}lll@{}}
\toprule
 & para-Krawtchouk Hamiltonian & Meixner model \cite{AJNW}\\
\midrule
lattice & finite, $N+1$ sites & $\mathbb N_0$, already infinite\\
polynomials & para-Krawtchouk (bi-lattice) & Meixner\\
the limit & $N\to\infty$: the system grows & $c\to1^-$: the realization degenerates\\
spectrum along it & $\{0,\gamma,2,2+\gamma,\dots\}$, rigid & $\nu+\mathbb N_0$, rigid\\
origin of $\sfrac{1}{\eta^2}$ & marginal impurity at the centre & representation label $\nu$ at the edge\\
limit operator & singular osc.\ on $\mathbb R$: $\mathcal D^+_{(2-\gamma)/4}\oplus\mathcal D^+_{(2+\gamma)/4}$
 & singular osc.\ on $\mathbb R_{>0}$: one $\mathcal D^+_\nu$\\
algebra & Hahn $\halg$ at finite $N$, $\su$ in the limit & $\su$ at every $c$\\
mirror, fractional revival & yes: $e^{-\ii\pi\N}=\xi I+\zeta'\R$ & absent\\
\bottomrule
\end{tabular}
\end{center}

Three of the contrasts deserve a word. The first is the  nature of the limit. $N\to\infty$
grows the lattice and produces the operator with the staggering $(-1)^n$, whereas $c\to1$ leaves the system untouched (same space, same spectrum, same
Casimir) and only degenerates the realization. Consistently, both limits leave the spectrum unchanged and act only on the eigenfunctions.

The second is the origin of the $\sfrac{1}{\eta^2}$ term. In the
para-Krawtchouk chain it comes from the impurity at the centre of the lattice and in the Meixner
model, from the representation label $\nu$ sitting at the edge $m=0$ of the half lattice.  The symmetrized couplings of
\eqref{eq:meixnerJ0} are proportional to $\sqrt{(m+1)(m+2\nu)}=m+\nu+\tfrac12
-\tfrac{(2\nu-1)^2}{8m}+O(m^{-2})$. 
Compared with the
values $\nu=\tfrac14,\tfrac34$, for which $\sigma(\sigma-1)=0$, the tail is
$-\sigma(\sigma-1)/(8m)$. 
Measured in the continuum variable, with $\eta=(n-j-1)/\sqrt{j+1}$ at the centre and
$\eta\approx\sqrt{\Delta m}$ at the edge, both tails are inverse-square in the distance to
the singular point. They are thus marginal in the sense of \S\ref{sec:eqlimit}, which is why
a finite $\sfrac{1}{\eta^2}$ coefficient survives in both cases.

The third and most telling one, is the algebra. A finite matrix cannot carry
$\su$ since all its non-trivial unitary irreducible representations are
infinite-dimensional. This is consistent with the bispectral pair of the
para-Krawtchouk model generating the quadratic Hahn algebra $\halg$, with $N$ in the
structure constants, $\su$ emerging only in the limit of \S\ref{sec:algebra}, and emerging
twice. On the
half-lattice the algebra is available from the outset: the model is an element of a fixed
$\mathcal D^+_\nu$, the parameter $c$ is a group boost the Casimir does not see, and the
continuum limit is a degeneration of the realization within one and the same
representation. That the lattice is ``already infinite'' is thus not a convenience but the
algebraic point itself. 

Note finally that in the $\mathfrak{sl}(2|1)$ oscillator of \cite{JVdJ}, whose position
wavefunctions are Meixner polynomials tending to the paraboson (Laguerre) ones, the two
$\su$ constituents have Bargmann indices $\beta/2$ and $(\beta+1)/2$ ($\beta>0$ the label of
the representation in \cite{JVdJ}). They differ by $\tfrac12$, which is exactly the
configuration $\nu_{\rm odd}-\nu_{\rm even}=\tfrac12$ that the para-Krawtchouk model attains
only at $\gamma=1$. For $\gamma\neq1$ the offset $\gamma/2$ obstructs any such superalgebra
pairing of the two channels, in accord with the breaking of the supersymmetry with
reflections noted in \S\ref{sec:FR}. The persistence of the symmetry algebra along the limit
on infinite lattices extends to higher rank: the two-dimensional discrete oscillator built
on bivariate Meixner polynomials keeps the $\sutwo$ symmetry of its continuum limit at every
scale \cite{GGLV}.

\subsection{Doubling}
\label{sec:doubling}

The combination of two orthogonal families into one mirror-symmetric operator is, when the
parameters of the two families differ by integers, the doubling of Oste and Van der Jeugt
\cite{OVdJ1}, who classify the ways in which two Hahn, dual Hahn or Racah families whose
parameters differ by integers can be made to obey a common pair of recurrence relations.
Each pair obtained in this way is a Christoffel--Geronimus pair, as is every quadratic
decomposition of an orthogonal family \cite{MP}, and the classification identifies the cases
in which the transform stays within the family. The finite oscillator obtained by doubling has its reflection inside the algebra. This
algebra is an extension of $\sutwo$ in which a reflection operator, commuting or
anticommuting with the generators, appears in the commutator of the two ladder operators.
This algebra is a special
case of the Bannai--Ito algebra $\bialg$, and the oscillator contracts to the parabosonic
(Wigner--Dunkl) oscillator \cite{OVdJ2}.

The para-Krawtchouk chain combines two families into one mirror-symmetric operator, but
in a different way. Its two dual Hahn families have the parameters $\alpha$ and $\beta$ interchanged, and the
offset between them is fixed by the continuous parameter $\gamma$ rather than being an
integer. This difference is visible in the limit, where the continuous offset gives the scalar term
$\tfrac{\gamma^2-1}{4\eta^2}$, whose two channels carry the Frobenius exponents
$\tfrac{1\mp\gamma}{2}$, whereas an integer offset gives the reflection-graded term
$\mu(\mu-\R)/\eta^2$ of the Wigner--Dunkl oscillator, whose exponents differ by one. The two agree at $\gamma=1$, where the para-Krawtchouk limit is the
ordinary oscillator, that is, the Wigner--Dunkl oscillator with $\mu=0$.

The doubling
framework and the reducibility of the four families (para-Krawtchouk, para-Racah,
$q$-para-Racah and para-Bannai--Ito) are taken up in \cite{BCLMNV}. The doubling procedure also
underlies the exactly solvable inhomogeneous Su--Schrieffer--Heeger models of \cite{CLMPV},
of which the para-Bannai--Ito Hamiltonian of \cite{PBIcont} is an instance for a
particular value of its parameters.

\section{Fractional revival and harmonic evolution in the limit}
\label{sec:FR}

An XX spin chain of $N+1$ sites
with nearest-neighbour couplings $\sqrt{U_n}$ and local magnetic fields $B_n$ conserves the number of
spins up.  On its one-excitation subspace, spanned by the states $\ket{n}$ in which
the single spin up sits at site $n$, its Hamiltonian acts as $J$ \cite{VZ,paraRacahChain}.
The transport of an excitation along the chain is given by the evolution
$e^{-\ii\theta J}$ on $\mathbb C^{N+1}$, with $\theta$ the time. Two properties of this evolution are of interest \cite{GVZ}. 
Perfect state transfer at time $T$ means that an excitation initially located at one end is found with
probability one at the other end
\begin{equation}
    e^{-\ii TJ}\ket{0}=e^{\ii\phi}\ket{N}\,.
\end{equation}  Fractional revival at time $T$ means that it
reappears as a coherent superposition localized at the two ends
\begin{equation}
\label{eq:FRdef}
e^{-\ii TJ}\ket{0}=\xi\,\ket{0}+\zeta\,\ket{N},\qquad |\xi|^2+|\zeta|^2=1,
\end{equation}
with probability $|\xi|^2$ of being found back at the origin and $|\zeta|^2$ of having been
transferred, perfect transfer is the case $\xi=0$.  Both properties can be read off the spectral data of
$J$. By \eqref{eq:eigvec}, since $p_0=1$ 
\begin{equation}
\ket{0}=\sum_s\sqrt{w_s}\,\ket{x_s} \,,
\end{equation} and
by \eqref{eq:eps}, 
\begin{equation}
\ket{N}=\Rc\,\ket{0}=\sum_s\epsilon_s\sqrt{w_s}\,\ket{x_s}\,,
\end{equation} hence
\eqref{eq:FRdef} holds if and only if
\begin{equation}
\label{eq:FRphases}
e^{-\ii Tx_s}=\xi+\zeta\,\epsilon_s\qquad\text{for every }s,
\end{equation}
that is, if and only if the phases $e^{-\ii Tx_s}$ take a single value $e^{\ii\phi_+}$ on the
eigenvalues with $\epsilon_s=+1$ and a single value $e^{\ii\phi_-}$ on those with
$\epsilon_s=-1$. Then 
\begin{align}
&\xi=\tfrac12(e^{\ii\phi_+}+e^{\ii\phi_-})\,,\\
&\zeta=\tfrac12(e^{\ii\phi_+}-e^{\ii\phi_-})\,,
\end{align}
and since $I=\sum_s\ket{x_s}\bra{x_s}$ and
$\Rc=\sum_s\epsilon_s\ket{x_s}\bra{x_s}$, \eqref{eq:FRphases} is the operator identity
\begin{equation}
\label{eq:FRop}
e^{-\ii TJ}=e^{\ii\phi}\big(\cos\vartheta\,I+\ii\sin\vartheta\,\Rc\big)=e^{\ii\phi}\,e^{\ii\vartheta \Rc},
\qquad \phi=\tfrac12(\phi_++\phi_-),\quad \vartheta=\tfrac12(\phi_+-\phi_-).
\end{equation}
At time $T$ every state, not only $\ket{0}$, becomes the superposition of itself and its
mirror image, with weights $\cos^2\vartheta$, $\sin^2\vartheta$. Perfect
state transfer corresponds to $\vartheta\equiv\tfrac\pi2\pmod\pi$. Since $\Rc$ is an
involution, the angle at the times $rT$, $r=1,2,\dots$, is $r\vartheta$. On a bi-lattice,
where $\epsilon_s$ is constant on each sub-lattice, the criterion is that $e^{-\ii Tx_s}$ be
constant on each of the two sub-lattices.

\subsection{Fractional revival on the linear bi-lattice.}
\paragraph{$N$ odd.}
For $N$ odd, $\epsilon_s=-1$ on the even sub-lattice $\{2s\}$ and $\epsilon_s=+1$ on the odd one
$\{2s+\gamma\}$. The criterion is met at $T=\pi$ for every $\gamma$, since
$e^{-\ii\pi x_{2s}}=1$ and $e^{-\ii\pi x_{2s+1}}=e^{-\ii\pi\gamma}$, that is, $\phi_-=0$ and
$\phi_+=-\pi\gamma$.

\begin{proposition}[Fractional revival]
\label{prop:FR}
For every odd $N$,
\begin{equation}
\label{eq:FRfinite}
e^{-\ii\pi J}=\xi\,I+\zeta\,\Rc
=e^{-\ii\pi\gamma/2}\Big(\cos\tfrac{\pi\gamma}2\,I-\ii\sin\tfrac{\pi\gamma}2\,\Rc\Big),\qquad
\xi=\frac{1+e^{-\ii\pi\gamma}}2,\quad \zeta=\frac{e^{-\ii\pi\gamma}-1}2,
\end{equation}
with revival angle $\vartheta=-\pi\gamma/2$ and phase $\phi=-\pi\gamma/2$. The return and
transfer probabilities are $|\xi|^2=\cos^2\tfrac{\pi\gamma}2$ and
$|\zeta|^2=\sin^2\tfrac{\pi\gamma}2$, and depend only on $\gamma$. In particular
$e^{-\ii\pi J}=-\Rc$ for $\gamma=1$, which is perfect state transfer with phase $-1$. 
\end{proposition}

The proof is \eqref{eq:FRop} with the two phases above. The revival angle is proportional to the
offset $\gamma$ of the bi-lattice and can take any value. This is in contrast with the
quadratic bi-lattice $\{(s+a)^2,(s+c)^2\}$ of the para-Racah Hamiltonian, with grid
parameters $a$ and $c$, on which the phases $e^{-\ii T(s+a)^2}$ and $e^{-\ii T(s+c)^2}$ align
only under a Diophantine condition on $a$ and $c$ \cite{paraRacahChain,PRcont}. 
Because the
spectrum lies on a bi-lattice, the evolution $e^{-\ii\theta J}$ is periodic for rational
$\gamma$ or quasi--periodic otherwise, i.e. the excitation performs a harmonic motion between the ends, and
$e^{-\ii\theta J}$ plays for the model the role of a discrete fractional Fourier transform, of
which the revival of Proposition~\ref{prop:FR} is the value at $\theta=\pi$. Any correct
continuum limit must preserve this structure; we shall see it does.

By \eqref{eq:XYlimit}, the continuum image of $J$ is the operator
$\N:=\tfrac12(\Hs-2+\gamma)$, and by \eqref{eq:genH} its eigenvalues are $2t$, $t=0,1,2,\dots$,
on the $\R$-even channel and $2t+\gamma$ on the $\R$-odd channel. The two channels thus give
the two Laguerre towers $\{0,2,4,\dots\}$ and $\{\gamma,2+\gamma,\dots\}$, which together form
the infinite bi-lattice with offset $\gamma$. The criterion above applies to $\N$ verbatim, with its eigenfunctions in place of the eigenvectors and $\R$ in place of $\Rc$. Since
$e^{-2\pi\ii t}=1$ on the even channel and $e^{-\ii\pi(2t+\gamma)}=e^{-\ii\pi\gamma}$ on the
odd one, with $P_\pm=\tfrac12(I\pm\R)$ the projectors on the two channels we obtain
\begin{equation}
\label{eq:FRcont}
e^{-\ii\pi\N}=P_++e^{-\ii\pi\gamma}P_-=\xi\,I+\zeta'\,\R,\qquad
\xi=\frac{1+e^{-\ii\pi\gamma}}2,\quad \zeta'=\frac{1-e^{-\ii\pi\gamma}}2=-\zeta.
\end{equation}
This is the value at $\theta=\pi$ of the fractional Fourier transform
$\mathcal F_\theta=e^{-\ii\theta\N}$ of the singular oscillator, and for $\gamma=1$ it is the
half-period. It coincides with the discrete identity \eqref{eq:FRfinite} with $\Rc$ replaced
by $-\R$, so the sign of the coefficient of the
mirror flips as it must be: the site reversal acts as $-\R$ on the envelopes
\eqref{eq:RcR}, the even sub-lattice being mirror--odd at finite $N$ and $\R$--even in the
continuum.

Read on wave functions, \eqref{eq:FRcont} says that at $\theta=\pi$ every state
becomes the coherent superposition of itself and its mirror image,
$\psi\mapsto\xi\psi+\zeta'\R\psi$, with weight
$|\xi|^2=\cos^2\tfrac{\pi\gamma}2$ in place and $|\zeta'|^2=\sin^2\tfrac{\pi\gamma}2$ at the
reflected position i.e.,  a packet localized on one side of the origin reappears as two copies, one
on each side. This is the continuum counterpart of the mirror pair $\ket{0}$, $\ket{N}$ and is fractional
revival in the sense of wave-packet dynamics \cite{AP,Robinett}, namely the reconstruction of a
packet as a finite superposition of copies, here two copies related by the reflection; the
pure mirror image ($|\zeta'|=1$, at $\gamma=1$) is the continuum image of perfect state
transfer. 

The entire harmonic evolution $e^{-\ii\theta J}$ of the discrete model thus passes to
$e^{-\ii\theta\N}$. For  $\gamma=1$, $\theta=\pi$ it is pure parity (the continuum image of
PST), and for $\gamma=\half$ it is, up to a global phase, $\tfrac1{\sqrt2}(I+\ii\R)$, a ``square
root of the reflection'' (the balanced, entangling revival). Fractional revival is thus a
property of the limiting Hamiltonian in its own right, read off its spectrum and its
reflection grading, and the reflection that carries it is the same object that grades the two
channels.

\begin{remark}[Supersymmetry with reflection]
Each channel of $\Hs$ factorizes separately: with
$\mathcal A_\sigma=\partial_\eta+\eta-\tfrac{\sigma}{\eta}$ acting in the channel of
exponent $\sigma$, one has $\Hs=\mathcal A_\sigma^\dagger\mathcal A_\sigma+(2\sigma+1)$
there, since $\sigma(\sigma-1)=\tfrac{\gamma^2-1}4$ for $\sigma=\sigma_\mp$: a
shape--invariant isotonic oscillator, the partner
$\mathcal A_\sigma\mathcal A_\sigma^\dagger+(2\sigma+1)$ having the centrifugal coefficient
$\sigma(\sigma+1)$. At
$\gamma=1$ the two channels are the even and odd sectors of the ordinary oscillator,
separated by $2$, and are paired by the supercharge $\tfrac1{\sqrt2}(\partial_\eta\R+\eta)$ of
supersymmetric quantum mechanics with reflections \cite{PVZ}, whose square is
$\tfrac12(\Hs-\R)$. For $\gamma\neq1$ the separation of the two towers is $2\gamma\neq2$ and
no supercharge can pair them (the partner of the channel of exponent $\sigma$ carries the
coefficient $\sigma(\sigma+1)$, which equals $\tfrac{\gamma^2-1}4$ only at $\gamma=1$): the
supersymmetry of \cite{PVZ} is broken by the impurity, the reflection $\R$ surviving only as
the $\mathbb Z_2$ grading of the two channels, and $\R$ is absent from the potential, so
there is no Dunkl supersymmetry either.
\end{remark}

\paragraph{$N$ even.}
For even $N$ the criterion is met at $T=\pi$ with
$e^{\ii\phi_+}=1$ on the even sub-lattice ($\epsilon=+1$) and $e^{\ii\phi_-}=e^{-\ii\pi\gamma}$ on
the odd one. Hence $\vartheta=+\pi\gamma/2$ and
\begin{equation}
\label{eq:FReven}
e^{-\ii\pi J}=\xi\,I+\zeta'\,\Rc,\qquad
\xi=\frac{1+e^{-\ii\pi\gamma}}2,\quad \zeta'=\frac{1-e^{-\ii\pi\gamma}}2,
\end{equation}
with the same return and transfer probabilities $\cos^2\tfrac{\pi\gamma}2$,
$\sin^2\tfrac{\pi\gamma}2$ as in Proposition~\ref{prop:FR} and the opposite sign in front of
the mirror. Since the site reversal is now $+\R$ on the envelopes, \eqref{eq:FReven} is
literally the continuum identity \eqref{eq:FRcont}, whereas for $N$ odd the two were
related by $\Rc=-\R$. Both parities therefore lead to the same
statement in the limit, $e^{-\ii\pi\N}=\xi I+\zeta'\R$, which is a property of the singular
oscillator.

\section{Conclusion}
\label{sec:concl}

The continuum limit of the para-Krawtchouk oscillator, the single Hamiltonian defined by the
tridiagonal Jacobi matrix, is the singular (isotonic) oscillator
$-\partial_\eta^2+\eta^2+\tfrac{\gamma^2-1}{4\eta^2}$.
We obtained this in two parallel ways:
by contracting the discrete Schr\"odinger equation, and by taking the limit directly on the
polynomials, using their explicit hypergeometric form and the reduction to two Hahn families.  For every degree it yields  the generalized Hermite functions \eqref{eq:genH}
(Proposition~\ref{prop:conv}),with the envelope read off the normalization
of the polynomials, exactly as one passes from Krawtchouk to Hermite.

The parameter $\gamma$ is preserved throughout, carried by a central impurity of the
couplings, decaying as the inverse square of the distance, the marginal rate, which survives as the $\sfrac{1}{\eta^2}$
term and vanishes at $\gamma=1$. For even $N$ the impurity is shared between the
couplings and the fields, and only their sum survives, the same for both parities
(\S\ref{sec:even}). 

The site reversal becomes the reflection
$\eta\mapsto-\eta$.  It commutes with the limiting
Hamiltonian and is absent from its potential, so the limit is the scalar singular oscillator
and not the Wigner--Dunkl parabosonic oscillator. It grades the two sub-lattice channels,
which correspond to the two self-adjoint extensions at the singular origin selected by the
discrete model and are separated in energy by $\gamma$. That separation is precisely the
fractional-revival phase, which becomes the exact operator identity
$e^{-\ii\pi\N}=\xi I+\zeta'\R$ for the limiting Hamiltonian. Algebraically, the bispectral
pair of the discrete model, which generates the quadratic Hahn algebra, becomes a pair of
generators of the $\su$ dynamical symmetry of the singular oscillator, $\mathbb Z_2$-graded by
the mirror. This places the model beside the Krawtchouk (harmonic) and Hahn (parabosonic)
oscillators.

The comparison with the $\su$ Meixner oscillator on the semi-infinite lattice
(\S\ref{sec:meixner}) shows the same singular oscillator reached along a road where no
contraction is needed, the algebra being present at every scale. It also identifies the
para-Krawtchouk continuum as two such half-line models, of complementary Bargmann indices
$\tfrac{2\mp\gamma}4$, related by the mirror that carries the revival. More broadly
(\S\ref{sec:doubling}), the para-Krawtchouk Hamiltonian combines two Hahn families in the
manner of a doubling, but with a continuous offset $\gamma$. The integer-offset doubles of
Oste and Van der Jeugt, whose reflection is algebraic and whose continuum limit is
parabosonic, meet the present construction at $\gamma=1$, where the para-Krawtchouk limit is
the ordinary oscillator.

The continuum limits of the
two other families are taken up in two companion papers: the para-Racah Hamiltonian on its
quadratic bi-lattice \cite{paraRacah,paraRacahChain} contracts to a P\"oschl--Teller box \cite{PRcont},
and the para-Bannai--Ito Hamiltonian, on the $q=-1$ grid, to the Scarf~I Hamiltonian, the square of a Dunkl supercharge, for an even number of sites, and to the matrix supersymmetric P\"oschl--Teller
pair for an odd number \cite{PBIcont}.

\section*{Acknowledgements}
NC thanks the Centre de Recherches Math\'ematiques (CRM) for its hospitality. LV is funded in
part through a Discovery Grant of the Natural Sciences and Engineering Research Council (NSERC)
of Canada; SZB, QL, LM and MJRV enjoy scholarships and fellowships provided by this fund.

\section*{Conflict of interest}
The authors declare that they have no conflict of interest.

\section*{Data availability statement}
No new data were created or analysed in this study.

\appendix
\renewcommand{\theequation}{A.\arabic{equation}}
\setcounter{equation}{0}
\section{The alternating ansatz for the discrete Schr\"odinger equation}
\label{app:contraction}

Unless stated otherwise, in this appendix, we will distinguish the quantities defined in \cite[Ch. II.1]{GantmacherKreinBook} from ours via an \textit{arc} (e.g., their Jacobi matrix will be denoted as $\wideparen{J}$ and ours simply as $J$). In \cite[Ch. II.1, eq. (1)]{GantmacherKreinBook},
they write:
\begin{equation}
    \label{JacobiGK}
    \wideparen{J} =\begin{bmatrix}
   \wideparen{a}_1 & -\wideparen{b}_1 & 0 & \cdots & 0\\
   -\wideparen{c}_1 & \wideparen{a}_2 & -\wideparen{b}_2  & \cdots & 0\\
   0 & -\wideparen{c}_2 & \wideparen{a}_3  & \cdots & 0\\
   \vdots & \vdots & \vdots  & \ddots & \vdots\\
   0 & 0 & 0 & \cdots & \wideparen{a}_{\wideparen{n}}\\

    \end{bmatrix}.
\end{equation}
We can therefore make the following identifications with our $J$: $B_{k-1}=\wideparen{a}_k$; $\sqrt{U_k}=-\wideparen{b}_k=-\wideparen{c}_k$; and $N=\wideparen{n}-1$. We can recognize that our Jacobi matrix is thus the symmetric case of the more general $\wideparen{J}$. Associated to this $\wideparen{J}$, they define the polynomials $\wideparen{D}_{k}$ as:
\begin{equation}
    \label{eq:polynomialsGK}
    \wideparen{D}_k\noleft(x\noright)= \det\noleft(\wideparen{J}_k-x\noright),
\end{equation}
with $\wideparen{J}_k$ being the $k^{\text{th}}$ leading principal submatrix of $\wideparen{J}$, and these satisfy:
\begin{equation}
    \wideparen{D}_k\noleft(x\noright)= \noleft(\wideparen{a}_k-x\noright)\,\wideparen{D}_{k-1}\noleft(x\noright)-\wideparen{b}_{k-1}\wideparen{c}_{k-1}\wideparen{D}_{k-2}\noleft(x\noright).
\end{equation}
Shifting $k$, re-arranging and multiplying throughout by $(-1)^{k}$, we can re-write this expression in terms of our $B_k$ and $U_k$ as:
\begin{equation}
    x\cdot\noleft(-1\noright)^{k}\,\wideparen{D}_{k}\noleft(x\noright)= \noleft(-1\noright)^{k+1}\,\wideparen{D}_{k+1}\noleft(x\noright)+B_{k}\cdot\noleft(-1\noright)^{k}\,\wideparen{D}_{k}\noleft(x\noright)+U_{k}\cdot\noleft(-1\noright)^{k-1}\,\wideparen{D}_{k-1}\noleft(x\noright).
\end{equation}
This is in the same form as \eqref{monicPKrecurrence}, from which we can therefore obtain that $\noleft(-1\noright)^{k}\,\wideparen{D}_{k}=P_{k}$, i.e., these polynomials correspond to our monic orthogonal polynomials up to an alternating negative sign. We can hence also translate property $7^\circ.$ of \cite[Ch. II.1]{GantmacherKreinBook} for our case. This can be restated as:
\begin{quote}
\textit{Considering the $\wideparen{n}$ ordered eigenvalues of $\wideparen{J}$, i.e., $\wideparen{\lambda}_1<\dots<\wideparen{\lambda}_{\wideparen{n}}$, then, at a given $\wideparen{\lambda}_k$, the sequence of $\wideparen{D}_0\noleft(\wideparen{\lambda}_k\noright),\dots,\wideparen{D}_{\wideparen{n}-1}\noleft(\wideparen{\lambda}_k\noright)$ has $k-1$ sign changes.}
\end{quote}
Applying this to our polynomials, we can first identify the eigenvalues $x_{s}=\wideparen{\lambda}_{s+1}$. Correspondingly, the appropriate sequence of $\wideparen{D}_k\noleft(x_s\noright)$ will have $s$ sign changes. Conversely, since $\noleft(-1\noright)^k$ changes sign $N$ times and $\noleft(-1\noright)^{k}\,\wideparen{D}_{k}=P_{k}$, overall the sequence of $P_{k}$ will have $N-s$ sign changes. This directly carries over to the orthonormal $p_{k}$'s, as these are related to the monic $P_{k}$'s via a positive normalization factor.

\end{document}